\documentclass[]{article}
\usepackage[utf8]{inputenc}
\usepackage[margin=1in]{geometry}
\usepackage{rotating}
\usepackage{booktabs}
\usepackage{array}
\usepackage{amsmath, amsfonts}
\usepackage{parskip}
\usepackage{color}
\usepackage[hidelinks]{hyperref}
\usepackage{bbm}
\usepackage{algorithm}
\usepackage[noend]{algpseudocode}
\usepackage{subcaption} 
\usepackage{enumitem}

\usepackage{mathtools}
\mathtoolsset{showonlyrefs}

\usepackage[authoryear,round]{natbib}

\usepackage{amsthm}
\newtheorem{theorem}{Theorem} 
\newtheorem{prop}{Proposition}

\newtheorem{definition}{Definition}
\newtheorem{remark}{Remark}

\usepackage{mathtools}

\def\Ical{\mathcal{I}}

\def\R{\mathbb{R}}
\def\E{\mathbb{E}}

\newcommand{\cost}{C}
\newcommand{\marketsize}{R}

\providecommand{\keywords}[1]
{
  \small	
  \textbf{\textit{Keywords---}} #1
}

\title{\vspace{-1.5cm} Principal-Agent Hypothesis Testing}
\author{Stephen Bates, Michael I.~Jordan, Michael Sklar, Jake A.~Soloff}
\date{\today}

\begin{document}

\maketitle

\begin{abstract}
Consider the relationship between a regulator (the \textit{principal}) and an experimenter (the \textit{agent}) such as a pharmaceutical company. 
The pharmaceutical company wishes to sell a drug for profit, whereas the regulator wishes to allow only efficacious drugs to be marketed. 
The efficacy of the drug is not known to the regulator, so the pharmaceutical company must run a costly trial to prove efficacy to the regulator. 
Critically, the statistical protocol used to establish efficacy affects the behavior of a strategic, self-interested agent; a lower standard of statistical evidence incentivizes the agent to run more trials that are less likely to be effective. 
The interaction between the statistical protocol and the incentives of the pharmaceutical company is crucial for understanding this system and designing protocols with high social utility. 
In this work, we discuss how the regulator can set up a protocol with payoffs based on statistical evidence. 
We show how to design protocols that are robust to an agent's strategic actions, and derive the optimal protocol in the presence of strategic entrants.
\end{abstract}


\section{Introduction}
\label{sec:intro}

Across government, academia and business, regulators rely on statistical tests to control access to influence and markets. Medical regulators test drugs, journal reviewers assess scientific claims, and managers in technology or finance companies evaluate system innovations. Researchers---pharmaceutical companies, scientists or software engineers---for their part, bear various upfront costs of designing and running their experiments. If their proposals are approved, they are personally rewarded, either financially or with prestige or other intangibles. The researchers are thus, effectively, taking a calculated bet on their proposal. 

Regulators, on the other hand, have a mandate to set up a trustworthy system that not only performs valid statistical inference but also delivers high social utility. 
In particular, if they disregard the researcher's profit motive, the regulator risks being exploitable.
That is, if it is profitable for a researcher to bluff---for a pharmaceutical company to submit an ineffective drug, a scientist to study a treatment with no effect, an engineer to submit an ineffective feature---then the regulator risks becoming overrun by low-value activity, and statistical errors are amplified, leading to low social utility. In this paper, we propose a decision-theoretic framework which explicitly accounts for researcher incentives, such that ineffectual proposals are not profitable. 

We consider as a motivating example the Food and Drug Administration (FDA), which regulates trillions of dollars of activity in the United States by controlling approval for medical treatments based on clinical trials. Here, commercial ventures first carry out background research to develop a candidate treatment, and then they must decide whether to sponsor a clinical trial typically costing millions of dollars.
If the drug is approved based on the results of the clinical trial, the company is rewarded with access to a market worth millions or even billions of dollars. Companies consider their financial calculus of whether to sponsor a clinical trial based on their private information, including the preliminary results of earlier research. 
Crucially, this financial calculus depends on the standard of evidence that the FDA implements---looser standards incentivize more clinical trials being run for less promising candidates. 
We explore in Section~\ref{sec-fda} how some FDA pathways are potentially exploitable for some high-profit drug categories.

By formalizing the bet the researcher is already taking, the regulator establishes an economic basis for statistical inference. Betting scores have recently resurfaced as a distinct and powerful approach to statistical inference~\citep[see, e.g.][]{shafer2021testing, ramdas2023game}. Our work builds on this line of work and connects game-theoretic statistical inference to contract theory, a body of work in the general area of economic mechanism design \citep{hart1987theory}. As a form of evidence, betting scores avoid some of the pathologies of significance testing, and by incorporating statistical evidence in an economic contract, we can afford a great deal of flexibility to researchers without ignoring their incentives. 

We model this interaction as a game between two players. The first player is known as the \textit{principal} (e.g., a regulator) and the second is called the \textit{agent} (e.g., a pharmaceutical company). 
The agent wishes to make money by selling a product, but must seek regulatory approval by the principal. 
The principal wishes to ensure that products on the market are legitimate.
We assume that the agent knows the quality of their product, denoted $\theta \in \Theta$ (referred to as a \emph{type} in the economics literature). Larger values indicate higher quality, with zero representing no improvement relative to a baseline. 
The principal, on the other hand, has no prior information about the quality~$\theta$.
To grant approval, the principal requires that the agent provide evidence for its product having sufficient quality; e.g., by conducting a randomized controlled trial (RCT) and testing a null hypothesis, $H_0 : \theta \le \theta_0$. The agent chooses whether to sponsor such a trial to seek regulatory approval based on their private information. 
In this work, the agent and the principal will enter into a contract that caps the reward the agent can receive as a function of the statistical evidence for the quality of the product. We explore how different contracts change the incentive landscape of the agent, and develop optimal contracts in this setting. 

\subsection{Incentives affect inference: a simple example}
\label{sec:fda_simple}
\label{sec:stories}

We begin with a stylized example to highlight the interaction between an agent's incentives and the principal's statistical protocol. Suppose there are two types of pharmaceutical companies: companies with ineffective drugs ($\theta = 0$) and companies with effective drugs ($\theta = 1$). Further, assume that the company knows its type, while the regulator does not. The company may choose to pay \$10 million to run a clinical trial, which results in a statistical test for the null hypothesis that the drug is ineffective. Suppose that the test is carried out so that it has 5\% type-I error and 80\% power to reject when $\theta = 1$; see below.

\begin{center}
    \begin{tabular}{|l||c|c|}
\hline
type & P(approve) & P(non-approve)  \\
\hline
\hline
    $\theta = 0$ & 0.05 & 0.95  \\
    \hline
    $\theta = 1$ & 0.80 & 0.20 \\ 
\hline
\end{tabular}
\end{center}

\vspace{.2cm}

Is this a good statistical protocol? The answer depends on how much money the pharmaceutical company will make, among other things. In particular, depending on the total profit the company earns when they are approved, even companies with ineffective drugs may be incentivized to run a trial.
\begin{itemize}
    \item \textbf{Case 1: small profit.} Suppose that companies who receive approval make \$100 million in profit, 10 times their investment. In this case, agents of type $\theta = 0$ would choose \emph{not} to run trials, since their expected profit for running a trial is -\$5 million. Hence, all approved drugs would be effective drugs.
    \item \textbf{Case 2: large profit.} Suppose that companies who receive approval make \$1 billion in profit, 100 times their investment. In this case, agents of type $\theta = 0$ \emph{would} choose to run trials: their expected profit from seeking approval is \$40 million. On average, 5\% of such agents would receive approval, so many of the approved drugs may actually be ineffective. For example, if there were 20 times as many ineffective drugs as effective drugs, a quick calculation shows that 56\% of approved drugs would be ineffective ones.
\end{itemize}
Clearly, the effectiveness of a statistical protocol depends on the incentives of the agents.

Conversely, the statistical protocol changes the incentives of the agents. Consider again the large profit case above, where agents receive 100 times their initial investment if they receive approval. Now, however, suppose the principal changes to a stricter protocol such that the probability of approval is only $0.005$ when $\theta = 0$. With this more stringent standard, the agents of type $\theta = 0$ are no longer incentivized to submit their drug for approval; they would lose money, on average, if they chose to run clinical trials for drugs they know to be ineffective. Thus, changing the statistical protocol will affect how the agents are incentivized to act.

In this work, we show how to design statistical protocols---which we term \emph{statistical contracts}---that account for this interplay. 

\subsection{Related work}
\label{sec:related_work}

Contract theory is the study of incentives when parties transact in the presence of private information; see~\citet{laffont2001, bolton2004contract}, and~\citet{salanie2005economics} for overviews. This is an extensive literature adjacent to information design~\citep[e.g.,][]{kamenica2011bayesian} and auction theory~\citep[e.g.,][]{hurwicz2006designing}, and we review only a narrow slice most related to statistical inference and computation. In particular, tools from algorithmic game theory have recently been applied to study contract theory; for example, \cite{babaioff2006combinatorial}~study the problem of incentivizing a collection of agents to perform a task requiring coordination among the agents. 
More recently, \citet{duetting2020complexity, guruganesh2021contracts}, and~\citet{castiglioni2022bayesian} show that finding optimal contracts is computationally hard. In parallel, several investigations have addressed how well simple contracts can approximate the optimal utility; \citet{dutting2019simple} show that the best linear contract is within a constant fraction
of the optimal contract, and \citet{carroll2015robustness} present complementary results on the robustness of linear contracts. Similarly,~\citet{alon2021contracts} characterize the set of implementable rules in a setting with both private actions and private information, showing that efficient algorithms are
possible in some cases. These works are all primarily computational; thus far little work has addressed the questions of learning or statistical inference in contract theory.  Applying contract theory to statistical inference, \citet{schorfheide2012use},~\citet{schorfheide2016hold}, and~\citet{spiess2018optimal} study the problem of delegating statistical estimation to a strategic agent, and~\cite{frazier2014incentivizing} consider the problem of incentivizing learning agents in a bandit setting. In the reverse direction, applying learning and inference to contract theory,~\cite{ke2009essays, ho2016adaptive} study the estimation and inference of an unknown distribution of agent types when a principal transacts with many agents from the distribution.

In view of our motivating example in Section~\ref{sec:stories}, we note that~\cite{min2023screening} studies a contract theory model for the FDA, where firms of different sizes choose to finance cheap or expensive trials. Our work focuses instead on how the regulator should conduct incentive-aware statistical inference in such regulator-firm relationships. In another direction, \cite{chassang2012selective} and \cite{ditillio2017} consider a model for randomized trials where an agent has a hidden action that affects the outcome. We note that in our model, there is no hidden action affecting the outcome.

Situating our work within contract theory, we study an \emph{adverse selection} model with a \emph{common value} structure in the principal's utility function. Our key departure from the usual adverse selection setup is that we do not assume that the principal has a prior distribution about the agents' hidden types. This makes statistical inference challenging. Our focus is on designing contracts that allow the principal to carry out statistical inference in order to properly assess the hidden types of the agents. Importantly, inference must be done in a way that takes into account the incentives of the strategic agents. Our work most closely relates to that of~\cite{tetenov2016economic}, who consider setting the type-I error level of a hypothesis test to account for an agent's payoffs. That work establishes a minimax protocol similar to our work for the case where the principal's action space involves setting the $p$-value threshold for approval, and it also analyzes the incentive structure of Phase III clinical trials for drug approval. See also~\citet{viviano2021when}, who study multiple hypothesis testing with incentives.

To carry out statistical inference in our economic setting, we use a statistical device called \emph{$e$-values}. An $e$-value for a null hypothesis $H_0$ is a nonnegative random variable $E$ such that $\mathbb{E}\left[ E \right] \le 1$ when~$H_0$ is true. Larger~$e$-values are then regarded as stronger evidence against the null. $e$-values are a broad class of model-comparison statistics which include Bayes factors and likelihood ratios. In contrast to more traditional $p$-values, $e$-values have a natural interpretation as a betting score, placing them at the interface of game theory and statistics. 
\cite{shafer2021testing} advocates for betting terminology as an effective framework for communicating results to a broad audience.
The betting analogy highlights another key benefit of game-theoretic statistics:
bets can be made sequentially, so evidence can be accumulated gradually within an experiment or aggregated across studies.
See~\cite{grunwald2020safe},~\citet{vovk2021values}, and~\citet{ramdas2023game} for more background on $e$-values. A key contribution of our work is the recognition that $e$-values are measures of statistical evidence that are well-suited for mechanism design.

\section{Incentive-aligned hypothesis testing}

\subsection{Setting}
\label{subsec:statistical-contracts-setup}
We assume the existence of a quality parameter, $\theta \in \Theta$, which is known to the agent. The principal does not know $\theta$, and has no information about the distribution
of values of $\theta$.

The principal moves first, offering a statistical contract of the following form:

\begin{center}
\fbox{
    \begin{minipage}{6.in}
    \vspace{.5em}
    \textbf{Statistical contract.} The agent follows these steps: 
        \begin{enumerate}[topsep=0pt,itemsep=-1ex,partopsep=1ex,parsep=1ex]
            \item Choose a license function $f$ from a \emph{menu} of options $\mathcal{F}$.
            \item Incur a cost $\cost$ to run a trial.
            \item Produce \emph{evidence} $Z$ from the trial, where $Z \sim P_\theta$.
            \item Receive a \emph{license} to make a profit at most $L = f\left( Z \right) \in [ 0, \infty )$.\vspace{.5em}
        \end{enumerate}
    \end{minipage}}
\end{center}

Based on their private information~$\theta$, cost~$\cost$ and menu of options~$\mathcal{F}$, the agent decides whether to opt in to the statistical contract. We let~$I\in\left\{ 0,1 \right\}$ denote the indicator that the agent opts in.

The principal receives utility $u\left( \theta, L \right)$
based on the true state~$\theta$ and the payoffs, normalized such that a utility of zero
corresponds to the principal's utility when the statistical contract is declined. 
The agent's profit from the contract is $L - \cost$ (or zero if they opt out), and we model the agent as seeking to maximize this profit; see below.
Our focus in this section is to design menus $\mathcal{F}$ that result in high utility for the principal. Formally, we assume that there is a space of possible license functions $\bar{\mathcal{F}}$, and the principal chooses a menu $\mathcal{F} \subset \bar{\mathcal{F}}$.

The principal utilizes a hypothesis testing framework.  We suppose there is a \emph{null} subset of types $\Theta_0 \subset \Theta$ and an \emph{nonnull} subset of types $\Theta_1 = \Theta \setminus \Theta_0$. The principal wishes to detect nonnull agents and transact with them, while avoiding null agents. We formalize this by taking the principal's utility function $u\left( \theta, L \right)$ to be nonnegative and nondecreasing in $L$ when $\theta \in \Theta_1$ but nonpositive and nonincreasing in $L$ when $\theta \in \Theta_0$. 
When the agent opts out, the principal's utility is zero. Putting this together, for a distribution $Q$ over agent types $\Theta$, the principal's expected utility is
\begin{equation}
\label{eq:principal_utility}
\E_{\theta \sim Q} \left[\E_{Z \sim P_\theta}\left[u\left( \theta, L \right) \cdot I \mid \theta\right]\right].
\end{equation}
Importantly, our conclusions about good statistical contracts do not depend on the precise specification of $u\left( \cdot, \cdot \right)$, other than the shape restrictions above. That is, we do not need to know the precise utility function $u\left( \cdot, \cdot \right)$ in practice in order to identify a good menu $\mathcal{F}$. Similarly, we do not assume that $Q$ is known to the principal when they design the statistical contract; instead, we derive statistical contracts that have high utility simultaneously for many distributions $Q$.

Turning to the agent, we model the agent as acting to optimize their profit license. That is, they respond to the principal's offer by choosing an element of $\mathcal{F}$ that results in maximum expected profit:
\begin{equation}
\label{eq:best_response_def}
    f^{\text{br}}\left(  \ \cdot \ ; \theta, \mathcal{F} \right) = \arg\max_{f \in \mathcal{F}} \left\{ \mathbb{E}_{Z \sim P_\theta}\left[f\left( Z \right)\right]\right\}.
\end{equation}
To elaborate on this notation, $f^{\text{br}}$ is any element of $\mathcal{F}$ that maximizes the expected profit for an agent with type $\theta$, so it depends on the underlying type $\theta$ and menu $\mathcal{F}$. 
If no license function has positive expected value, the agent opts out.
For technical ease, we assume that the functions $f \in \bar{\mathcal{F}}$ have finite expectation for all $\theta \in \Theta$ and that $\mathcal{F}, \bar{\mathcal{F}}$ are closed so that maxima such as in~\eqref{eq:best_response_def} exist and are attained. In practice, agents operate with a nonlinear utility function, instead of directly maximizing expected profit as in~\eqref{eq:best_response_def}. We will see that treating the agent's utility as linear is conservative in the case where the principal does not know the agent's utility function---see Remark~\ref{rmk:utility} below.

\begin{remark}[Relationship with Stackelberg equilibria]
With the above assumption, the principal's task of designing the optimal menu is related to the task of finding a Stackelberg equilibrium; the principal moves first and proposes a statistical contract, and then the agent best responds. However, our agent has hidden information (their type~$\theta$), and the principal's utility depends on this hidden information. Had we assumed a prior distribution over the hidden information, the statistical contract design problem would be to find a Stackelberg equilibrium in a Bayesian game. In this work, we do not assume knowledge of such a prior, and instead find a maximin optimal menu by protecting against the worst-case prior.
\end{remark}

\begin{remark}[Relationship with standard contract theory]
This setup is an example of \textit{adverse selection}, but---unlike most work---we do not assume that the principal has prior information about the distribution of the agent type $\theta$. Thus, the principal will need to employ statistical inference to take reasonable actions in the face of uncertainty. See Section~\ref{sec:related_work} for more discussion of how this relates to other results in contract theory.
\end{remark}

\subsection{Incentive-aligned statistical contracts}
Null values of $\theta$ correspond to agents
that do not meet the principal's quality requirements; 
the principal does not wish to offer a large
profit license to such agents. 
As we seek to design menus $\mathcal{F}$ with high utility, we will be particularly interested in the contracts 
such that null products are not profitable:
\begin{definition}[Incentive-aligned contract]
\label{def:incentive-align}
A statistical contract with menu $\mathcal{F}$ and cost~$\cost$ is \emph{incentive-aligned} if for all $\theta \in \Theta_0$ and all~$f\in \mathcal{F}$, $\E_\theta\left[ f\left( Z \right) \right] \le \cost$.
\end{definition}
An incentive-aligned statistical contract aligns the interests of the principal and agent so that agents
do not make a profit from ineffective products. 
A second interpretation is that such  statistical contracts deter bluffing agents; 
an agent that had certainty that their product was ineffective would 
choose to decline the statistical contract. We emphasize that the incentive-aligned property may not hold for naive regulatory mechanisms; see our previous example in Section~\ref{sec:fda_simple}. 

\begin{remark}[Nonlinear agent utility]\label{rmk:utility}
    Suppose that instead of employingn the best-response function~\eqref{eq:best_response_def}, the agent maximizes their expected utility as a function of profit
    \begin{equation}\label{eq:best_response_nonlinear}
    f^{\text{br}}\left(  \ \cdot \ ; \theta, \mathcal{F} \right) = \arg\max_{f \in \mathcal{F}} \left\{ \mathbb{E}_{Z \sim P_\theta}\left[\nu\left( f\left( Z \right) - \cost \right)\right]\right\},
    \end{equation}
    where $\nu : \R \to \R$ is a nondecreasing, concave utility function, normalized so that $\nu\left( 0 \right) = 0$. The agent opts out if no license function has positive expected utility. If a statistical contract with menu~$\mathcal{F}$ and cost~$\cost$ is incentive-aligned in the sense of Definition~\ref{def:incentive-align}, then by Jensen's inequality we have $\mathbb{E}_{Z \sim P_\theta}\left[\nu\left( f\left( Z \right) - \cost \right)\right] \le 0$ for all $\theta\in \Theta_0$. This shows null agents with concave utility will opt out, but it is possible that some nonnull agents opt out as well.
\end{remark}

We now characterize the set of incentive-aligned statistical contracts in terms of $e$-values, which are a tool from frequentist hypothesis testing.
\begin{definition}[$e$-value]
\label{def:evalue}
Consider a random variable $Z \in \mathcal{Z}$. A function $g : \mathcal{Z} \to \mathbb{R}_{\ge 0}$ is an \emph{$e$-value} for the null hypothesis $\{\theta \in \Theta_0\}$ if for all $\theta_0 \in \Theta_0$ we have
\begin{equation*}
    \mathbb{E}_{Z \sim P_{\theta_0}} \left[ g\left( Z \right) \right] \le 1.
\end{equation*}
\end{definition}
$e$-values are a measure of evidence against the null hypothesis---large $e$-values indicate more evidence against the null hypothesis. The reader should think of $e$-values as an alternative to $p$-values: both measure strength of evidence against the null hypothesis, but on a different scale. $e$-values can be converted to $p$-values and vice versa; see~\cite{vovk2021values}. We will let $\mathcal{E}$ denote the set of $e$-values.

The set of $e$-values is exactly the set of possible license functions in an incentive-aligned statistical contract after rescaling; this is a straightforward consequence of Definitions~\ref{def:incentive-align} and~\ref{def:evalue} above. We record this formally below.
\begin{prop}[Characterization of incentive-aligned statistical contracts]
\label{prop:$e$-values-ia-contract}
A menu $\mathcal{F}$ yields an incentive-aligned statistical contract if and only if $\mathcal{F} / \cost \subset \mathcal{E}$, where
$\mathcal{F} / \cost = \{f / \cost: f \in \mathcal{F}\}$ denotes the set of all functions in $\mathcal{F}$ scaled by $\cost^{-1}$.
\end{prop} 

Thus, the license functions that yield incentive-aligned statistical contracts are in
one-to-one correspondence with $e$-values. This simple observation makes a key, explicit connection between agent incentives and statistical hypothesis testing. One useful consequence is that the designer can select among the many $e$-values that have been catalogued in the literature as convenient building blocks.

\subsection{Optimality of the incentive-aligned statistical contract}
Incentive-aligned statistical contracts maximize the principal's worst-case utility
in a sense that we now make precise. Note that when the principal is a benevolent regulator whose utility is social welfare, this means that the incentive-aligned statistical contract has the best worst-case social welfare.
The idea is that the
principal needs only to screen out agents who know that their product is ineffective, which the incentive-aligned statistical
contract does. Moreover, a larger menu that is incentive-aligned is more attractive to the agents, so to maximize participation the principal should offer the largest incentive-aligned menu possible. We now turn to the formal presentation.

\subsubsection{Notation and goal}
We consider the case in which there are an infinite number of agents, each with their own $\theta \in \Theta$. Formally, we have a distribution $Q$ over $\Theta$ that gives the frequency of different
types of agents.
Recall the setting from Section~\ref{subsec:statistical-contracts-setup}: each agent wishes to make a profit, and they have a choice whether to enter into the statistical contract offered by the principal. 
We suppose that the principal receives 
bounded positive utility for true positives: $a_1 \ge u\left( \theta_1, L \right) \ge 0$ for $L \ge 0$, and that $u\left( \theta_1, L \right)$ is nondecreasing in $L$ when $\theta_1 \in \Theta_1$.
On the other hand, the principal receives negative utility for false positives: 
$u\left( \theta_0, L \right) < 0$ for all $L > 0$, and $u\left( \theta_0, L \right)$ is nonincreasing in $L$ for $\theta_0 \in \Theta_0$.
The principal seeks to optimize her utility in~\eqref{eq:principal_utility} by selecting the best menu of contracts $\mathcal{F} \subset \bar{\mathcal{F}}$ that will be offered to the agent. 

\subsubsection{Maximin optimality}
The principal's expected utility depends on the distribution over agent types, $Q$, and the menu offered, $\mathcal{F}$. The principal controls the latter, but may not know the former. Thus, we will seek a menu that performs well for many distributions $Q$. Our main result in what follows is that a statistical contract is maximin optimal if and only if it is incentive-aligned. Thus, a principal who does not know the distribution of agent types should select an incentive-aligned statistical contract.

To formalize this point, we consider a notion of maximin optimality. Consider menus of contracts $\mathcal{F}$
that are contained in some maximal set of functions $\bar{\mathcal{F}}$. For example, $\bar{\mathcal{F}}$ 
could be all measurable functions from $\mathcal{Z}$ to $\left[ 0,\marketsize \right]$.
We say a menu $\mathcal{F}$ 
is \textit{maximin optimal} if
\begin{equation}
\label{eq:maximin}
    \min_{Q} \E\left[\E\left[u\left( \theta, f^{\text{br}}\left( Z; \theta, \mathcal{F} \right) \right) \mid \theta\right]\right] = \max_{\mathcal{F}' \subset \overline{F}} \min_{Q} \E\left[\E\left[u\left( \theta, f^{\text{br}}\left( Z; \theta, \mathcal{F}' \right) \right)\mid \theta\right]\right],
\end{equation}
where we recall that $f^{\text{br}}$ is the agent's best response (given his knowledge of $\theta$) when offered a certain menu.
In words, a maximin statistical contract is one such that when nature picks the worst distribution $Q$ of agent types for our menu, it does as well as possible. We are now ready to state the primary result of this section.
\begin{theorem}
\label{thm:maximin}
A menu $\mathcal{F}$ is maximin optimal if
and only if it is incentive-aligned.
\end{theorem}
This result together with Proposition~\ref{prop:$e$-values-ia-contract} says that a menu is maximin optimal if and only if each license function in the menu
is an $e$-value (rescaled by $\cost$).
This result has broad implications for the design of statistical protocols within contracts.
If we want our statistical contracts to perform well across many cases, they must be incentive-aligned. 

To build intuition for this result, consider the special case where $Q = \pi_0 \cdot \delta_{\theta_0} + \left( 1-\pi_0 \right)\cdot \delta_{\theta_1}$ is a mixture of $\theta_0$ and $\theta_1$ with mixing proportion $\pi_0 \in \left[ 0,1 \right]$, where $\theta_0 \in \Theta_0$ is a null type and $\theta_1 \in \Theta_1$ is a nonnull type.
We will consider the case where there are a large number of null agents by
analyzing the behavior of a menu $\mathcal{F}$ as $\pi_0 \to 1$.
Intuitively, for any population with a large fraction of nulls, if the nulls are not all
screened out, then the utility will decrease as $\pi_0 \to 1$; an increasing fraction of agents opting in to the statistical contract will be null. On the other hand,
if all nulls are screened out, then the utility will go to zero as $\pi_0 \to 1$, since
all null agents will decline the statistical contract. Formally, we have the following:

\begin{prop}[Only incentive-aligned statistical contracts work with many nulls]
\label{prop:only-ia-contracts-feasible}
In the setting above, a menu $\mathcal{F}$ yields $\lim\inf_{\pi_0\to 1} \E\left[ u\left( \theta, L \right) \right] \ge 0$ for all $\theta_0 \in \Theta_0$ if and only if $\mathcal{F}$ corresponds to an incentive-aligned statistical contract.
\end{prop}

We thus see that if there are potentially many null agents who would take advantage of a poorly formulated statistical contract, the principal must select an incentive-aligned statistical contract or be exposed to negative utility. This idea at is the heart of Theorem~\ref{thm:maximin}.

\subsubsection{The best menu of incentive-aligned contracts}
We have shown that all incentive-aligned statistical contracts are maximin optimal. This result is coarse, however,
since there are many such menus $\mathcal{F}$ that result in an incentive-aligned statistical contract.
Which menu is best? We address this next. Our main result says
if the utility function $u\left( \theta, L \right)$ is linear in $L$, then
the menu of all license functions that are incentive-aligned is optimal. 

\begin{theorem}[Optimality of the menu of all $e$-values under linear utility]
\label{prop:opt_linear_util}
Let $\mathcal{F} = \cost \cdot \mathcal{E} = \{f : E_{Z \sim P_{\theta_0}}\left[ f\left( Z \right) \right] \le \cost \text{ for all } \theta_0 \in \Theta_0\}$, and let $\mathcal{F}'$ be any other menu such that the resulting statistical
contract is incentive-aligned. Assume $u\left( \theta_1, L \right)$ is affine in $L$ for $\theta_1 \in \Theta_1$, and let $Q$ be any distribution over $\Theta$.
Then, the menu $\mathcal{F}$ has higher utility:
\begin{equation*}
\E_{\theta \sim Q}\left[ \E\left[ u\left( \theta, f^{\text{br}}\left( Z; \theta, \mathcal{F} \right) \right) \mid \theta \right] \right] \ge \E_{\theta \sim Q}\left[ \E\left[ u\left( \theta, f^{\text{br}}\left( Z; \theta, \mathcal{F}' \right) \right) \mid \theta \right] \right].
\end{equation*}
\end{theorem}

Intuitively, this result arises from the fact that for incentive-aligned statistical contracts only nonnull agents are active. Nonnull agents maximizing their license value are also maximizing the principal's utility when it is linear in the license value. Therefore, the principal should give the agents as much room to optimize their license value as possible; i.e., the principal should give the agent the maximal menu possible---the menu of all rescaled $e$-values.

A utility function $u\left( \theta, L \right)$ that is linear in $L$ is natural in some cases. In our FDA example, $L$ is the profit the sponsor is able to earn from the drug's sale, which should increase according to the drug's effectiveness and scale of distribution. The public health benefit created by a drug may thus scale roughly linearly with profitability, and hence the utility function can be modeled as linear. In this case, the preceding result shows that the menu of all rescaled $e$-values is optimal, in a strong sense.

To shed light on the case in which $u\left( \theta, L \right)$ is not linear in $L$, we make the following observation. The menu of all rescaled $e$-values is the largest menu resulting in an incentive-aligned statistical contract, by Proposition~\ref{prop:$e$-values-ia-contract}. It then follows that, among incentive-aligned statistical contracts, it results in the largest number of nonnull agents choosing to enter the statistical contract, and it results in the highest expected license value for the agents. (Because both agent participation and expected license value increase for larger menus.) We record this formally next.

\begin{prop}
\label{prop:max_participation}
Let $\mathcal{F} = \cost \cdot \mathcal{E}$, and let $\mathcal{F}'$ be any other menu such that the resulting statistical
contract is incentive-aligned. Let $Q$ be any distribution over $\Theta$.
Then, the menu $\mathcal{F}$ results in fewer agents opting out
and a larger expected market size:
\begin{equation*}
\E_{\theta \sim Q}\left[ \E\left[ f^{\textnormal{br}}\left( \cdot ; \theta, \mathcal{F} \right) \mid \theta \right] \right] \ge \E_{\theta \sim Q}\left[ \E\left[f^{\textnormal{br}}\left( \cdot ; \theta, \mathcal{F}' \right) \mid \theta \right] \right].
\end{equation*}
\end{prop}

Even when the utility function is not linear in the second argument, we expect that agent participation and expected market size will be reasonable proxies of total utility, so the previous result suggests that the menu of all rescaled $e$-values will have high utility.

These results, together with Proposition~\ref{thm:maximin}, give practical guidance for the designer---if she does not have information about the prior distribution over $\theta$, she can deploy the menu of all incentive-aligned contracts. This menu is maximin optimal, and it is the best among the maximin optimal statistical contracts when the utility is linear in the second argument.

Note that although this menu is infinite, it can still be easily implemented. Since it is simple to verify whether or not a contract is incentive-aligned, the principal asks the agent for their proposed incentive-aligned contract and then proceeds with that contract, provided it is indeed incentive-aligned. For the agent offered the menu of all incentive-aligned contracts, we now derive the optimal license function for the agent (i.e., the element of the menu that he should select).

\subsection{Multiple rounds}

The connection of incentives with $e$-values also allows us to handle a multi-round interaction between the principal and agent. In each step, the agent can invest money to collect more data, with evidence accumulating. In Appendix~\ref{sec:dynamics}, we show how to design incentive-aligned statistical contracts in this multiple-round setting.

\section{The agent's best response}
\label{sec:best_evalue}

In this section, we derive the agent's best response~\eqref{eq:best_response_def} when the principal offers the menu~$\mathcal{F}$ of all rescaled $e$-values. We will show that that the agent's best response is to implement the classical likelihood ratio test. Turning to the details, suppose there is some maximal amount of profit~$\marketsize > \cost$ that the agent could make. Secondly, suppose that the null set is a singleton: $\Theta_0 = \left\{ \theta_0 \right\}$.
Then, an agent of type $\theta_1$ with linear utility seeks to maximize
\[
\E_{Z\sim P_{\theta_1}}\left[f\left( Z \right)\land \marketsize\right] \qquad \textnormal{subject to} \qquad f\in \mathcal{F}. 
\]
Noting~$\mathcal{F} = \left\{ f : \E_{Z \sim P_{\theta_0}}\left[ f\left( Z \right) \right] \le \cost \right\}$ and rescaling~$\varphi\coloneqq\frac{1}{\marketsize}f$, the agent equivalently seeks to maximize
\[
\E_{Z\sim P_{\theta_1}}\left[\varphi\left( Z \right)\land 1\right] \qquad \textnormal{subject to} \qquad \E_{Z \sim P_{\theta_0}} \left[\varphi\left( Z \right) \right] \le \frac{\cost}{\marketsize}. 
\]
The rescaled function~$\varphi : \Omega \to \left[ 0, 1 \right]$ represents the {\sl critical function} in a hypothesis test comparing the simple null~$\theta = \theta_0$ against the simple alternative~$\theta = \theta_1$. The agent's problem of finding the optimal contract is thus equivalent to the problem of designing a hypothesis test~$\varphi : \Omega \to \mathbb{R}$ to maximize power, $\E_{P_{\theta_1}}\varphi$, subject to the constraint on the Type I error, $\E_{P_{\theta_0}}\varphi\le \cost / \marketsize$. The celebrated Neyman--Pearson lemma~\citep{neyman1933ix} states that the optimal test is given by thresholding the likelihood ratio
\[
\varphi^* = \Ical \left\{\frac{\textnormal{d}P_{\theta_1}}{\textnormal{d}P_{\theta_0}} > \lambda\right\} \textnormal{ a.s. } P_{\theta_0},
\]
where the threshold $\lambda > 0$ chosen such that $\E_{P_{\theta_0}} \left[\varphi^* \right] = \cost / \marketsize$. Hence, the agent's best response is
$f^{\textnormal{br}} = \marketsize\varphi^*$.
We record this conclusion formally.
\begin{prop}[Neyman--Pearson lemma]
\label{prop:best_evalue}
Suppose $\Theta_0 = \{\theta_0\}$. The optimal incentive-aligned contract~$f^{\textnormal{br}}$ maximizing the expected license cap, $\E_{Z\sim P_{\theta_1}}\left[f\left( Z \right)\land \marketsize\right]$, for an agent of type $\theta_1 \ne \theta_0$ is given by
\[
f^{\textnormal{br}}\left( z \right) = \begin{cases}
\marketsize & \textnormal{if } \frac{\textnormal{d}P_{\theta_1}}{\textnormal{d}P_{\theta_0}}\left( z \right) > \lambda, \\
0 & \textnormal{otherwise,}
\end{cases}
\]
where~$\lambda > 0$ is chosen so that $\E_{Z\sim P_{\theta_0}} f^{\textnormal{br}}\left( Z \right) = \cost$.
\end{prop}
This result shows that, when choosing among all incentive-aligned contracts, an agent with~$\theta_1\ne \theta_0$ that maximizes their expected license cap will employ an all-or-nothing license function: the agent's license cap is either zero or~$\marketsize$, depending on the outcome of the trial, corresponding to the likelihood ratio test with type~I error constraint~$\cost / \marketsize$.

\subsection{A simple example}\label{sec:simple_example}

In this section, we illustrate the preceding results in the simple case where the outcome of the trial is normally distributed: 
\[
Z\sim P_\theta = \mathcal{N}\left( \theta, 1 \right).
\]
We take the null set to be $\Theta_0 = \left\{ 0 \right\}$, and the alternative to be $\Theta_1 = \left\{ 1 \right\}$.
We compare the optimal incentive-aligned statistical contract, $\mathcal{F} = \cost \cdot \mathcal{E}$, to the \emph{status quo} system: a standard statistical test that does not take incentives into account.
In particular, the status quo system is a statistical contract with a single license function,
\[
f^\textnormal{sq}\left( Z \right) =\begin{cases}
\marketsize & \textnormal{if } \overline\Phi\left( Z \right) < 0.05, \\
0 & \textnormal{otherwise},
\end{cases}
\]
where~$\overline\Phi\left( Z \right)$ is a~$p$-value based on a one-tailed~$Z$-test. The condition~$\overline\Phi\left( Z \right) < 0.05$ is equivalent to~$Z > 1.64$.  
Turning to the incentive-aligned statistical contract, when offered the menu of all $e$-values, by Proposition~\ref{prop:best_evalue} the agent will select the following license function:
\[
f^{\textnormal{br}}\left( Z \right) =\begin{cases}
\marketsize & \textnormal{if } \overline\Phi\left( Z \right) < \frac{\cost}{\marketsize}, \\
0 & \textnormal{otherwise}.
\end{cases}
\]

We next compare these two policies. There are two regimes. First, if~$\cost / \marketsize < 0.05$, the status quo license~$f^\textnormal{sq}$ is not incentive-aligned. Thus, strategic agents with $\theta = 0$ may still choose to run clinical trials.
Second, if~$\cost / \marketsize > 0.05$, the status quo license is incentive-aligned, but it is more conservative than necessary. In particular, it is more conservative than the agent's preferred license function $f^{\textnormal{br}}$. Thus, fewer drugs would be approved than with the optimal incentive-aligned statistical contract.

We define the principal's utility in terms of the costs and benefits to future patients. First, if the agent either opts out of the trial or obtains a license value~$f\left( Z \right) = 0$, then the drug does not go on the market, and future patients are unaffected by the outcome of the trial, in which case the principal's utility is zero. If a null drug enters the market, society incurs a cost of~$c_1 < 0$ due to financial burden, side effects, etc. Finally, if a nonnull drug enters the market, patients are spared of some fraction of the disease burden, resulting in positive utility~$c_2 > 0$ measured in potential years of life lost per patient. We use realistic values for these costs from~\citet{isakov2019fda}, considering two regimes: high-severity and low-severity. For the former, we have $c_2 / \left| c_1 \right| = 10$, and for the latter $c_2 / \left| c_1 \right| = 4 / 7$.

Turning to the results, we report the pincipal's expected utility in Figure~\ref{fig-sim}.
Here, we consider a distribution $Q$ over $\Theta$ with a fraction $\pi_0$ of nulls, i.e., $Q\left( \theta = 0 \right) = \pi_0$, for various values of the parameter $\pi_0$. The behavior differs in the low-profit and high-profit regimes, as explained next. The left panel shows a low-profit case, where the market cap~$\marketsize$ is five times the trial cost~$\cost$. Since~$\cost / \marketsize > 0.05$, the agents with~$\theta=0$ opt out of the trial in both systems, but the incentive-aligned menu has a higher utility because it uses a less conservative cutoff. That is, a larger fraction of the effective drugs are approved.  The right panel shows a high-profit setting, where~$\marketsize$ is fifty times the trial cost~$\cost$. Under the status quo system, null agents have an incentive to enter the trial. Thus, when $\pi_0 \approx 1$, the principal's utility is negative, since many null agents enter the trial and $5\%$ of them are approved by chance. On the other hand, Theorem~\ref{thm:maximin} implies that incentive-aligned system has higher worst-case utility, which here occurs when $\pi_0 = 1$.

\begin{remark} (The agent's best response in a sequential experiment).
    In Appendix~\ref{sec:dynamic-programming}, we study the agent's best response when evidence can be accumulated over $T$ rounds of an experiment. The agent can recruit and test subjects in batches and stop the experiment at any stage, for any reason, exiting with the current total license value. We apply dynamic programming to show how the agent can adaptively choose a license at each round to optimize their final license value. Even when the agent has a linear utility, their best response is no longer to place an all-or-nothing bet in the first round. In a simulation study, we compare to the single-round case in the simple Gaussian example above, finding that the multi-round agent can leverage cost savings to consistently achieve a comparable and in some cases significantly higher overall profit. 
\end{remark}

\begin{figure}[t!]
	\centering
	\centerline{
	\includegraphics[width=.49\textwidth]{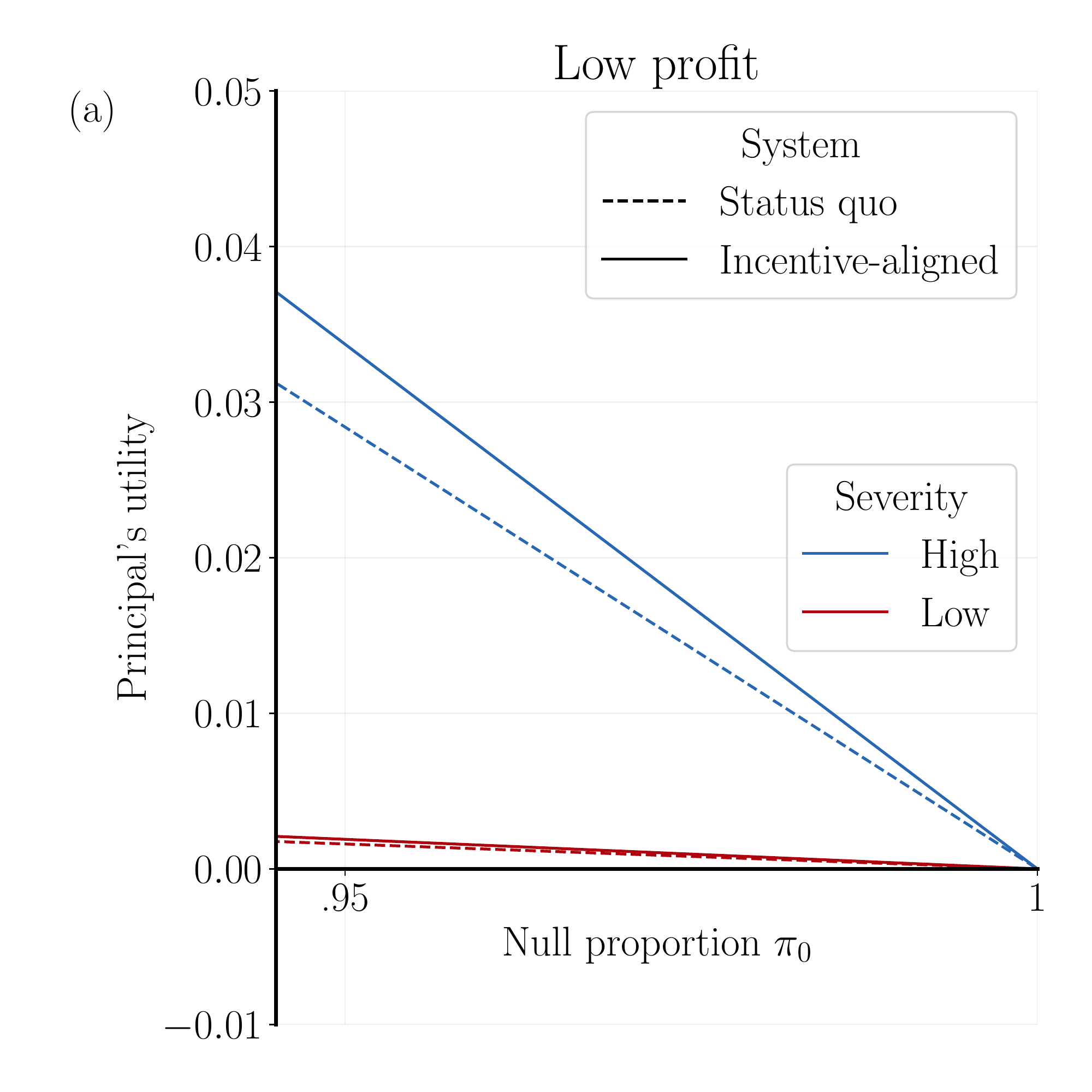}
	\includegraphics[width=.49\textwidth]{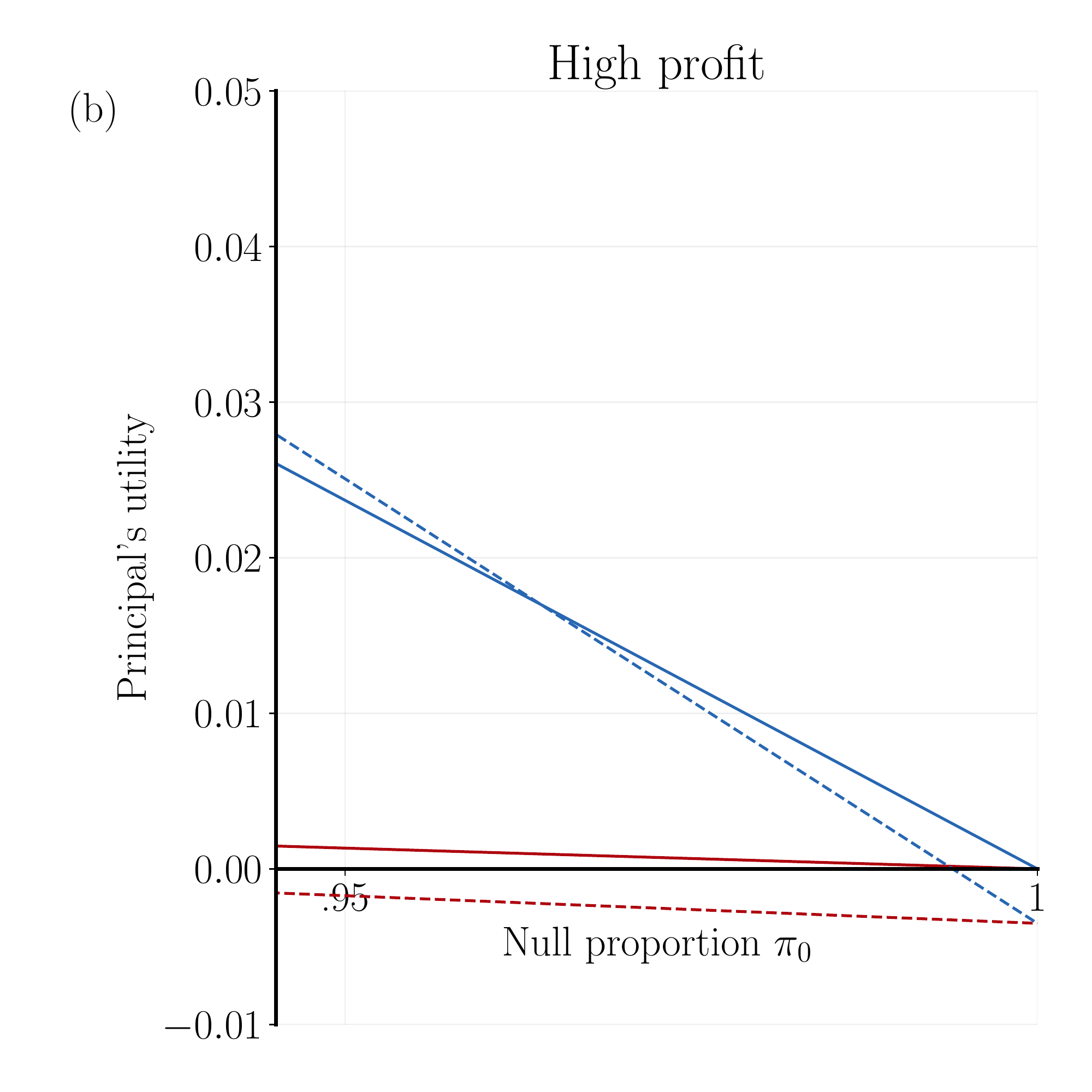}
	}
    \caption{The expected utility that the principal derives from a single agent.  
    (a) When $\marketsize / \cost = 5$, it is not profitable for null agents to enter the status quo trial, so the overall welfare is comparable to that of the incentive-aligned mechanism, except that the latter has slightly higher power. (b) When $\marketsize / \cost = 50$, null agents enter the status quo trial but not the incentive-aligned mechanism. For $\pi_0\approx 1$, the incentive-aligned mechanism has higher welfare. }\label{fig-sim}
\end{figure}

\section{Incentives and the FDA}\label{sec-fda}

In this section, we discuss the functioning of the FDA via the lens of a principal-agent problem, and the deterrence constraints this imposes on the regulator. We will analyze the current FDA protocol to see if the expected profit of a placebo drug is positive (i.e., not incentive-aligned), negative (i.e., incentive-aligned), or approximately zero. Of course, the FDA does not always follow a formal algorithm for evaluating clinical trials; it retains significant discretion in decision-making, and the standards for efficacy can vary by the treatment area or reviewing team \citep{janiaud2021us}. Yet, we can study a simplification based on its written guidelines as a point of reference.

Before a new drug can be marketed, the FDA requires a sponsor to run confirmatory clinical trials to establish the drug's safety and efficacy. In assessing this evidence for approval, FDA regulators retain significant flexibility of interpretation, but it is typically expected to provide either two confirmatory trials with positive results (often interpreted as level $p<.05$ on a two-sided test, with the result in the correct direction), or one confirmatory trial (typically multi-center, well-controlled, with strong results) in conjunction with reasoning to support a determination of ``substantial evidence.'' 
According to \cite{morant2019characteristics}, from 2012-2016 all non-orphan, non-oncology drugs which were approved with a single pivotal trial achieved $p$-values of less than $0.005$.
The study of \cite{janiaud2021us} ``[suggests] that the FDA does not have a consistent implicit or operational standard for defining `substantial evidence' in contexts
when flexible criteria are used.'' 
The FDA also has an accelerated approval pathway, where a drug may be released under less rigorous standards of evidence, based on proxies for the primary outcome of interest, and sometimes also with lower efficacy standards. After accelerated approval, the FDA may continue to review the drug for final approval as further evidence accumulates---during which time the drug is released to the market and may be sold for profit. The FDA does not cap or regulate prices and profits.
See \cite{food2019demonstrating,haslam2019confirmatory,janiaud2021us} for more information on when single versus multiple trials are used.

To facilitate formal analysis, we analyze three simplified statistical protocols that track the FDA evaluatory process. First, we consider a protocol that requires significance in two independent trials---we call this the \emph{standard} protocol. The probability that a placebo drug is approved with this protocol is $0.025^2 = 0.000625$. Second, we consider the case where only one trial is required, but at the $0.01$ level for a two-sided test. We call this the \emph{modernized} protocol, in view of the FDA Modernization Act of 1997\nocite{modernization_1997}.
We assume the probability that a placebo drug is approved with the modernized protocol is $0.005.$ Lastly, we consider a protocol where two trials are conducted and approval is granted if either trial yields a two-sided $p$-value below $0.05$, so that the probability of approving a placebo is $0.0494$. We call this the \emph{high-discretion accelerated} protocol. This protocol is looser than typical FDA behavior, but in extreme cases approvals with this low level of statistical evidence do occur. Notably, aducanumab (Aduhelm) was approved after two trials, one with a statistically significant result and one with a statistically insignificant result~\citep{aduhelmreport}.

Estimates for the cost of Phase III trials vary widely. \citet{moore2018estimated} estimate a median of 20 million in direct costs for a Phase III trial, but notes that costs can vary 10x in either direction. \citet{moore2020variation} estimate median total costs for pivotal trials (which may include 2 trials) at $\$ 50 $ million. \cite{dimasi2016innovation} estimate average total Phase III costs as $\$255$ million out-of-pocket and $\$314$ million accounting for opportunity cost of capital, although this analysis is based on a private dataset and may focus on large companies with large trials and higher costs figures \citep{love2019criticism}. \cite{wouters2020estimated} estimate mean out-of-pocket costs at $\$291$ million  among trials whose costs were broken out in published SEC filings, although this condition may systematically exclude cheaper trials. \cite{schlander2021much} provides a review. 
A sponsor aiming to bluff the system by submitting a known placebo to a clinical trial would prefer to use cheap, small trials, in order to minimize cost and maximize variance. Thus, for our analysis, we use a value of $\$50$ million for total Phase III costs. 

Various estimates place the average total capitalized cost of developing a single new successful drug at $\$161$ million - $\$4.54$ billion \citep{schlander2021much}. In equilibrium, expected profits of each approved drug should be least as large as costs. In addition, there is a long tail of drug profitability, and there exist blockbuster drugs with sales of over $\$ 100$ billion~\citep{elmhirst2019biopharma}. 

\begin{table}[t]
\centering
\small
\begin{tabular}{|l|l|l|l|l|}
\hline
\multicolumn{1}{|c|}{Protocol} & \multicolumn{1}{c|}{$\mathbb{P}$(null approval)} & \multicolumn{1}{c|}{\begin{tabular}{@{}c@{}}Profit \\  if approved\end{tabular}} & \multicolumn{1}{c|}{\begin{tabular}{@{}c@{}}Expected value \\  of placebo\end{tabular}} & \multicolumn{1}{c|}{Incentive-aligned?} \\
\hline
\hline
standard & 0.000625 & \$1B & -\$49M  & Yes \\
standard & 0.000625 & \$10B & -\$44M &  Yes \\
standard & 0.000625 & \$100B &  \$13M &  \textbf{No} \\
\hline
modernized & 0.005 & \$1B &  -\$45M & Yes  \\
modernized & 0.005 & \$10B &  \$0M & \textbf{Borderline} \\
modernized & 0.005 & \$100B & \$450M & \textbf{No} \\
\hline
accelerated & 0.0494 & \$1B  & \$-1M & \textbf{Borderline}  \\
accelerated & 0.0494 & \$10B & \$444M & \textbf{No} \\
accelerated & 0.0494 & \$100B & \$4.9B  & \textbf{No} \\
\hline
\end{tabular}

\caption{Incentive-alignment of three statistical protocols for varying drug market values, assuming a Phase III cost of $\$50$ million.}
\label{tab:FDA_realistic}
\end{table}

We report on the expected value of a placebo drug under the three protocols above in Table~\ref{tab:FDA_realistic}. We find that for typical drugs with \$1-10B profit if approved, the standard protocol requiring two trials is incentive-aligned. For extremely profitable drugs earning \$100B or more, the protocol ceases to be incentive-aligned. Furthermore, as the protocol is loosened, it ceases to be incentive-aligned even for less profitable drugs. With the high-discretion accelerated protocol, incentive alignment is lost even for a typical drug. In the blockbuster scenario, an organization may have an expected value of hundreds of millions or billions of dollars of profit for running a clinical trial on a placebo.

The primary implication of this analysis is that if the standard of evidence required by the FDA is loosened, it may cease to be incentive-aligned for the more profitable drugs. The right standard of evidence for the FDA is a source of ongoing debate, and some call for much looser protocols. For example, the Bayesian decision analysis of~\cite{isakov2019fda} prompts the authors to call for thresholds from 1\% to 30\%, depending on the class of drug. We emphasize that their analysis does not consider how the incentive landscape is impacted by adopting this looser standard of evidence. In particular, in view of Table~\ref{tab:FDA_realistic}, we worry that greatly loosening the standard of evidence may incentivize
clinical trials for unpromising candidates, resulting in too many false positives.

\cite{tetenov2016economic} pursued a similar analysis of Phase III incentives, and concluded that a Type I error of up to 15\% is incentive-aligned for an average drug. In contrast, when we consider unusually low Phase III costs and unusually profitable drugs, we find that the FDA could already be at risk of violating incentive-alignment in rare cases. 

\subsection{Limitations}

There are important limitations to the above analysis. In particular, our calculation omits additional regulatory checks against approving ineffective drugs and punishments for agents who intentionally run clinical trials for drugs they believe to be ineffective. These considerations include: additional evidence standards that the FDA may impose, the reluctance of insurers to compensate generously for a drug with marginal evidence of efficacy, lawsuits and liability costs, and the risk of a company developing a negative reputation among consumers, insurers, and regulators.

A second major limitation to our work is that our mathematical investigation considers the case where agent has perfect knowledge of the drug's efficacy. In reality, the agents (pharmaceutical companies) seldom have this knowledge. Despite this, there is a seemingly infinite supply of abandoned drug candidates. The results in this work (e.g., Theorem~\ref{thm:maximin}) imply that a statistical protocol which fails to align incentives may encourage a company to run clinical trials for many of these; thus, we believe that regulators' protocols should at the very least be incentive-aligned.
An expanded model accounting for the agent's lack of perfect knowledge would give clarity about how much more strict (beyond incentive-aligned) the protocol should be for optimal utility. 

\section{Discussion}

Our study of hypothesis testing in the principal-agent model has forged connections between statistical inference and ideas from the economic theory of mechanism design. The primary conclusion of this work is that the principal who does not know the distribution of agent types should deploy the menu of all $e$-values to create an incentive-aligned statistical contract, protecting the principal from null strategic agents. Instead of controlling the Type I error rate at some arbitrary level, our framework bases the standard of evidence on the researcher's incentives to guarantee incentive-aligned hypothesis testing. This concrete mechanism is practical and adapts to the unknown agent distribution and unknown market size. The key observation underlying this work is that $e$-values are a measure of statistical evidence that are on the right scale for reasoning about incentives; accordingly, we expect that this tool will be broadly useful for statistical inference in the presence of strategic behavior.

In the remainder of this section we consider various possible extensions to the current work.

\paragraph{The menu of all $e$-values adapts to unknown market size.}
We have focused on the case where the agent's private information is his own type $\theta$. Our setup automatically handles an additional piece of private information---market size. For example, a pharmaceutical company may know that a drug would result in a profit at most $\marketsize$, due to their market research, but the FDA may not know this. Formally, we suppose the agent's utility is $\min\left( L,\marketsize \right)$ for a cap $\marketsize$ that is unknown to the principal. 

The menu of all $e$-values extends to this setting without modification. In the single-round setting, the principal offers menu $\mathcal{F} = \cost \cdot \mathcal{E}$, and then the agent will respond by choosing the optimal license function as derived in Section~\ref{sec:best_evalue}.  This has the following natural statistical interpretation. Proposition~\ref{prop:best_evalue} shows that the best statistical protocol when the max payoff is $\marketsize$ is the Neyman--Pearson test at a certain level that depends on problem parameters---the investment amount $\cost$ and the cap $\marketsize$. The above considerations show that the principal doesn't need to know $\marketsize$; the statistical contract is set up in such a way that the self-interested agent will choose to run Neyman--Pearson test at the right level for the problem parameters. Similarly, the principal does not know the right likelihood ratio statistic because they do not know which alternative the agent has. Even though the principal lacks this information, the menu of all $e$-values is a mechanism that incentivizes the agent to run the Neyman--Pearson test at the appropriate level, given their private information of their type $\theta$ and the market size $\marketsize$.

\begin{figure}[t]
	\centering
	\centerline{
	\includegraphics[width=.5\textwidth]{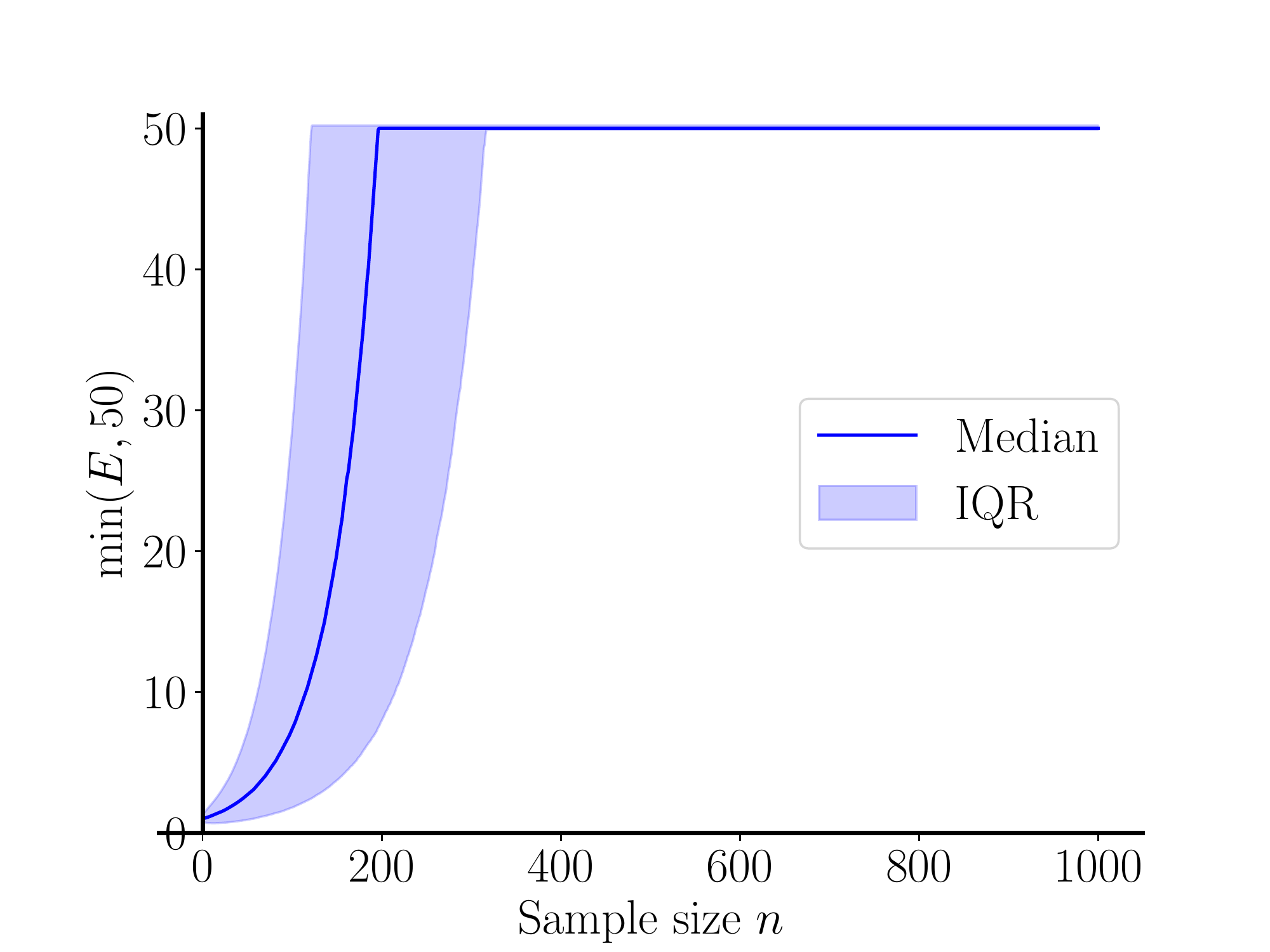}
	}
    \caption{ $e$-value $E = \exp\left( \theta_1 \sum_{i=1}^n Z_i - \frac{n\theta_1^2}{2} \right)$ as a function of~$n$, where~$\theta_1 = 0.2$. Under the alternative~$Z_i\stackrel{\textnormal{iid}}{\sim}\mathcal{N}\left( .2, 1 \right)$, the $e$-value grows exponentially, quickly reaching the market cap even in a high profit market such as in the right panel of Figure~\ref{fig-sim}. 
    }
    \label{fig-sample-size-sim}
\end{figure}

\paragraph{Choosing a sample size.} 
Another situation that can be handled with a variant of our framework is when the agent chooses the amount of statistical evidence to gather by varying their amount of investment. For example, a pharmaceutical company might select the sample size of their clinical trial based on their knowledge of the effect size. Formally, we would encode this by giving the agent a set of investment levels, $\cost \in \left\{ \cost_1,\dots,\cost_m \right\}$, with statistical evidence that depends on the investment level (e.g., on the sample size): $Z \sim P_{\theta, \cost}$. For example, if $\cost_1$ corresponds to a sample size of $100$ and $\cost_2$ corresponds to a sample of size $1000$, we might have $P_{\theta, \cost_1}$ be $100$ samples from a $\mathcal{N}\left( \theta,1 \right)$ distribution and $P_{\theta, \cost_2}$ be $1000$ samples from a $\mathcal{N}\left( \theta,1 \right)$ distribution.

This setting is of practical interest, so we briefly discuss the considerations underlying the choice of sample size. First, we note that $e$-values typically grow exponentially in the sample size. To see this, note that one way to construct $e$-values for a sample of size $n$ is to multiply together one $e$-value for each sample. For example, if $Z_i \sim \mathcal{N}\left( \theta, 1 \right)$ for $i=1,\dots,n$, then one possible $e$-value is
$
    \prod_{i=1}^n \exp\left( Z_i - 1/2 \right),
$
since $\E\left[ \exp\left( Z_i - 1/2 \right) \right] = 1$ under the null hypothesis that $\theta \le 0$.
This means that the marginal value of one additional sample is growing with sample size.
If the cost of statistical evidence grows linearly with the sample size, this means that the agent with utility function $\min\left( L,\marketsize \right)$ would either (a) decline to run a trial or (b) choose a sample size and $e$-value such that the profit license is nearly $\marketsize$ with large probability. We demonstrate this with a simulation in Figure~\ref{fig-sample-size-sim}; the $e$-value grows exponentially in the sample size until it reaches the cap $\marketsize$, so an agent deciding the sample size would choose a value near the beginning of the plateau. Thus, our proposal to limit agent profits based on the license value~$L = f\left( Z \right)$ need not limit profits of nonnull agents. 

\section*{Acknowledgements}

We thank Nivasini Ananthakrishnan, Jon McAuliffe and Aaditya Ramdas for helpful discussions.  This work was supported in part by the Vannevar Bush Faculty Fellowship program under grant number N00014-21-1-2941 and in part by the European Union (ERC-2022-SYG-OCEAN-101071601).

\bibliographystyle{chicago}
\bibliography{incentives}

\newpage

\appendix
\section{Proofs}
\label{sec:proofs}

\begin{proof}[Proof of Proposition~\ref{prop:$e$-values-ia-contract}]
Suppose $f \in \cost \cdot \mathcal{E}$. Notice the net profit the agent receives is $f\left( Z \right) - \cost = \cost \cdot \left( E - 1 \right)$ for the $e$-value $E = f\left( Z \right) / \cost$. Under $H_0$, $\E\left[ \cost \cdot \left( E - 1 \right) \right] \le 0$ by the defining property of the $e$-value $E$.

For the converse, suppose we have an incentive-aligned statistical contract with menu $\mathcal{F}$. For any $f \in \mathcal{F}$, consider the random variable $f\left( Z \right) / \cost$. This random variable is nonnegative with expectation at most 1, since otherwise the contract would not be incentive-aligned. Thus, it is an $e$-value, so the contract can be written in the form above.
\end{proof}

\begin{proof}[Proof of Proposition~\ref{prop:only-ia-contracts-feasible}]
Suppose $\mathcal{F}$ yields an incentive-aligned statistical contract. Then, all null agents opt out of the statistical contract. Thus
\begin{equation*}
    \mathbb{E}\left[ u\left( \theta, L \right) \right] = \mathbb{E}\left[ u\left( \theta, L \right) \cdot \Ical\left\{ \theta = \theta_1 \right\} \right] \ge 0,
\end{equation*}
where the inequality follows from the fact that the principal's utility is nonnegative when $\theta = \theta_1$.

Conversely, suppose $\mathcal{F}$ is not incentive-aligned. Then, for some $\theta_0 \in \Theta_0$ and $f \in \mathcal{F}$, we have that $\mathbb{E}_{Z \sim P_{\theta_0}}\left[ f\left( Z \right) \right] \ge \cost$. In particular, $P_{\theta_0}\left( f\left( Z \right) > 0 \right) > 0$. Thus, $\mathbb{E}_{Z \sim P_{\theta_0}}\left[ u\left( \theta_0, f\left( Z \right) \right) \right] < 0$ by our assumption on $u\left( \cdot, \cdot \right)$. Note that the principal's expected utility is
\begin{equation*}
    \mathbb{E}\left[ u\left( \theta, L \right) \right] = \pi_0 \cdot \mathbb{E}_{Z \sim P_{\theta_0}}\left[ u\left( \theta_0, f\left( Z \right) \right) \right] + \left( 1 - \pi_0 \right) \cdot \mathbb{E}_{Z \sim P_{\theta_1}}\left[ u\left( \theta_1, f\left( Z \right) \right) \right].
\end{equation*}
Since the positive utility is bounded above, we see that as $\pi_0 \to 1$, this quantity converges to a negative number. This proves the claim.
\end{proof}

\begin{proof}[Proof of Theorem~\ref{thm:maximin}]
For any $\mathcal{F}$, notice that the expected utility is nonpositive for $Q$ supported only on $\theta_0$, by our assumption on the utility function $u\left( \cdot, \cdot \right)$, so the worst-case utility of a maximin statistical contract is zero. But the utility of incentive-aligned statistical contracts is always nonnegative; since null agents opt out of incentive-aligned statistical contracts, we have that the principal's expected utility is
\begin{equation*}
    \mathbb{E}\left[ u\left( \theta, L \right) \right] = \mathbb{E}\left[ u\left( \theta, L \right) \cdot \Ical\left\{ \theta \in \Theta_1 \right\} \right] \ge 0.
\end{equation*}
We conclude that incentive-aligned statistical contracts are maximin.

Conversely, Proposition~\ref{prop:only-ia-contracts-feasible} implies that
only incentive-aligned statistical contracts are maximin. 
\end{proof}

\begin{proof}[Proof of Theorem~\ref{prop:opt_linear_util}]
\begin{align*}
\E_{\theta \sim Q}\left[ \E\left[ u\left( \theta, f^{\text{br}}\left( Z; \theta, \mathcal{F} \right) \right) \mid \theta \right] \right] 
    &= \E_{\theta \sim Q}\left[ \E\left[g_0\left( \theta \right) + g_1\left( \theta \right) \cdot f^{\text{br}}\left( Z; \theta, \mathcal{F} \right) \mid \theta \right] \right] \\
    &= \E_{\theta \sim Q}\left[ \E\left[g_0\left( \theta \right) + g_1\left( \theta \right) \cdot f^{\text{br}}\left( Z; \theta, \mathcal{F} \right) \mid \theta \right] \cdot \Ical\left\{ \theta \in \Theta_1 \right\} \right] \\
    &= \E_{\theta \sim Q}\left[ \left( g_0\left( \theta \right) + g_1\left( \theta \right) \cdot \E\left[  f^{\text{br}}\left( Z; \theta, \mathcal{F} \right) \mid \theta \right] \right) \cdot \Ical
    \left\{\theta \in \Theta_1 \right\} \right],
\end{align*}
for some functions $g_0$ and $g_1$ by the assumption that $u\left( \theta, L \right)$ is affine in $L$ when $\theta \in \Theta_1$. The second equality comes from the assumption that $\mathcal{F}$ is incentive-aligned.

The final expression above is larger for $\mathcal{F} = \mathcal{F}_1$ than from $\mathcal{F} = \mathcal{F}_2$ when $\mathcal{F}_1 \supset \mathcal{F}_2$, by the definition of the agent's best response function $f^{\text{bf}}\left( \cdot; \theta, \mathcal{F} \right)$; see~\eqref{eq:best_response_def}. By Proposition~\ref{prop:$e$-values-ia-contract}, $\mathcal{F}$ in the theorem statement contains any incentive-aligned $\mathcal{F}'$. The claim follows.
\end{proof}

\begin{proof}[Proof of Proposition~\ref{prop:max_participation}]
It is clear that larger menus result in more participation.

For the second part, 
\begin{equation*}
\E\left[ f^{\textnormal{br}}\left( \cdot ; \theta, \mathcal{F} \right) \mid \theta \right] \ge \E\left[ f^{\textnormal{br}}\left( \cdot ; \theta, \mathcal{F}' \right) \mid \theta \right].
\end{equation*}
for any $\theta \in \Theta_1$ by the definition of the best response function $f^{\text{br}}$. Note that
\begin{equation*}
\E_{\theta \sim Q}\left[ \E\left[ f^{\textnormal{br}}\left( \cdot ; \theta, \mathcal{F} \right) \mid \theta \right] \right] = \E_{\theta \sim Q}\left[ \E\left[ f^{\textnormal{br}}\left( \cdot ; \theta, \mathcal{F} \right) \mid \theta\right] \cdot \Ical\left\{ \theta \in \Theta_1 \right\} \right]
\end{equation*}
and similarly for $\mathcal{F}'$, since $\mathcal{F}$ and $\mathcal{F}'$ are incentive-aligned. The result follows.
\end{proof}

\begin{proof}[Proof of Proposition~\ref{prop:$e$-values-ia-contract_multi}]
For the first result, let $N\left( 0 \right) = 0$ and define $N\left( t \right)$ for $t=1,\dots,T$ as in~\eqref{eq:net_profit_mgale}. Then, $N\left( t \right)$ for $t=0,\dots,T$ is a supermartingale when $\theta \in \Theta_0$.
The contract exits with profit $L\left( \tau \right)$ for some stopping time $\tau$, so the result follows.

For the converse, if the menu contains update functions that are not $e$-values, then for some $t$, some $\theta_0 \in \Theta_0$ and some $f_t$, $\E_{Z \sim P_{\theta_0}}\left[ f_t\left( Z \right) \right] > 1$. Thus the strategy that runs the trial only at stage $t$ and makes no withdrawals has expected profit
\begin{equation*}
    \cost_t \cdot \E_{Z \sim P_{\theta_0}}\left[ \left( f_t\left( Z \right) - 1 \right) \right] > 0
\end{equation*}
for an agent of type $\theta_0$. Thus, the statistical contract is not incentive-aligned.
\end{proof}

\begin{proof}[Proof of Theorem~\ref{thm:maximin_multi}]
The proof is analogous to that of Theorem~\ref{thm:maximin}.
\end{proof}

\begin{proof}[Proof of Proposition~\ref{prop:opt_linear_util_multi}]
The proof is analogous to that of Theorem~\ref{prop:opt_linear_util}
\end{proof}

\section{Dynamics}
\label{sec:dynamics}

We now consider the setting in which statistical evidence is gathered in stages. At each step, 
the agent can choose to withdraw from their profit license and/or to invest
money to gather more evidence.
For example, a clinical trial could proceed by first recruiting 100 subjects, and
then proceeding with batches of 100 subjects until either the profit cap is
sufficiently large or the company chooses to abandon the trial. It is possible
to develop incentive-aligned statistical contracts in this more complex setting using martingales to create an evolving profit cap that cannot be abused by a 
strategic agent.

\subsection{Profit license in multiple-round setting}

Let us now move to the multi-round setting. Now consider: what is the space of possible policies for the principal? It's very large, and must include notions such as ``allow the agent to choose any test with $p < .025$, and then reap $X$ of market profits." We consider the class of sequential games in discrete time.

\begin{center}
\fbox{
    \begin{minipage}{6.in}
    \vspace{.5em}
    \textbf{Multi-round Sequential Game} 
    \\
    For $t = 1,2\dots,T$:
        \begin{enumerate}[topsep=0pt,itemsep=-1ex,partopsep=1ex,parsep=1ex]
            \item The agent is allowed to withdraw, with no further payments.
            \item If the agent continues, the agent receives wealth $W\left( t \right)$, within a range of choices which may depend on all prior states of the game. ($W\left( t \right)$ would be negative if it is a payment to continue the experiment, and positive for a profitable payout.)
            \item Evidence $X\left( t \right)$ is generated, which may depend on $W\left( t \right)$, the value of $\theta$, and all prior states of the game.
        \end{enumerate}
    Under the convention $W\left( 0 \right) = 0$, the agent ends the trial with wealth $\sum \limits_{t=1}^T W\left( t \right)$, 
    \vspace{.5em}
    \end{minipage}}
\end{center}

In a principal-agent game, the principal may place constraints on the actions or reward available to the principal. For example, the principal (a regulator) can effectively constrain the Type I error rate of a design proposal of the agent (a drug sponsor), by preventing the agent from earning any positive rewards if the evidence standard is not met. Subject to these constraints, the sponsor may actively manage its clinical trials costs, such as early-exiting with futility stopping, increasing investment with adaptive sample sizing. It may also accrue some minor profits while trials are in progress.

We now consider the following class of contracts, which we shall call ``profit licenses."

\begin{center}
\fbox{
    \begin{minipage}{6.in}
    \vspace{.5em}
    \textbf{Multi-round Profit License} \vspace{.2cm}\\ 
    The initial license value is $L\left( 0 \right) = 0$. \\
    For $t = 1,2\dots,T$
        \begin{enumerate}[topsep=0pt,itemsep=-1ex,partopsep=1ex,parsep=1ex]
            \item The agent chooses to  withdraw
            $P\left( t \right) \in \left( 0, L\left( t-1 \right) \right)$ from their license, or continue without a withdrawal.
            \item The agent may pay $\cost_t$ to run stage $t$ of the trial, in which case:
            \begin{enumerate}[topsep=0pt,itemsep=-1ex,partopsep=1ex,parsep=1ex]
            \item The agent selects an update function $f_t$ from the menu $\mathcal{F}_t$.
            \item The trial produces evidence $Z_t \sim P_\theta$.
            \item The license value is updated as $L\left( t \right) = \left( L\left( t-1 \right) + \cost_t - P\left( t \right) \right)\cdot f_t\left( Z_t \right)$.
            \end{enumerate}
            \item Otherwise, license value is updated as $L\left( t \right) = L\left( t-1 \right) - P\left( t \right)$.
        \end{enumerate}
    The agent ends the trial with license value $L\left( T \right)$. \vspace{.5em}
    \end{minipage}}
\end{center}

In our FDA-pharma example, 
the agent is permitted to sell the drug to make $P\left( t \right)$ of profit at intermediate stages $t$, and ends the trial with a remaining license value of $L\left( T \right)$ of future profits. More generally, this setup allows the agent to either withdraw from their profit cap or invest in order to gather more statistical evidence in a dynamic fashion. This rich interaction offers great flexibility to the agent, while accounting for complex situations involving sequential data collection, adaptive trial sizing, and profit-taking. 

It might look like the class of such games $G$ is extremely large and, by comparison, the class of profit licenses is too restrictive. It turns out the class of profit licenses is sufficiently rich to include every incentive-aligned protocol which is exitable. Therefore the profit license characterizes the set of incentive-aligned policies within this seemingly large set of sequential interactions.

As with the single-round setting, we will show that the principal can use $e$-values to create incentive-aligned multi-round contracts. To discuss this precisely, we will need some additional notation. We let $I_t \in \left\{ 0, 1 \right\}$ be the indicator of whether the agent chooses to run the trial in stage $t$.
We let $L = L\left( T \right)$ be the license value at termination, $\cost = \sum_{t=1}^T \cost_t \cdot I_t$ be the total investment the agent made, and $P = \sum_{t=1}^T P\left( t \right)$ be the total withdrawal from the license that took place during the trial.
For any set of menus $\mathcal{F}_1,\dots,\mathcal{F}_T$, we saw that an \emph{agent strategy} is a protocol for choosing actions at each step. Formally, this is a collection of mappings from all the information until time $t-1$ (i.e., $P\left( 1 \right),\dots,P\left( t-1 \right)$, $I_1,\dots,I_{t-1}$, $Z_1,\dots,Z_{t-1}$, and $f_1,\dots,f_{t-1}$) to the actions at time $t$: the amount to withdraw $P\left( t \right)$, whether or not to run the trial $I_t$, and the choice of the update function $f_t \in \mathcal{F}_T$. For any set of menus $\mathcal{F}_1,\dots,\mathcal{F}_T$, and agent type $\theta \in \Theta$, we let
${\textsc{br-strategy}}\left( \theta ; \mathcal{F}_1,\dots,\mathcal{F}_T \right)$ be an agent strategy that maximizes $\E\left[ L + P - \cost \right]$.
Analogous to the single-round case, the principal's utility function $u\left( \theta, L + P \right)$ is a function of the agent type $\theta$ and the total profit $L + P$, with the shape constraints considered previously in Section~\ref{subsec:statistical-contracts-setup}.

With this notation, we say that a multi-round statistical contract is \emph{incentive-aligned} if for any agent with null type $\theta_0 \in \Theta_0$ and any agent strategy, we have,
\begin{equation}
\label{eq:multistage_ia_def}
    \E\left[ L + P - \cost \right] \le 0. 
\end{equation}
The term within the expectation should be understood as the agent's profit: it is their future profit $L$ plus their withdrawals $P$ minus their total investment $\cost$. As previewed above, $e$-values lead to incentive-aligned contracts:

\begin{prop}[Characterization of incentive-aligned multi-round contracts]
\label{prop:$e$-values-ia-contract_multi}
Suppose $\mathcal{F}_t \subset \mathcal{E}$ for $t = 1,\dots,T$. Then, the statistical contract is incentive-aligned as in~\eqref{eq:multistage_ia_def}.
Conversely, if $\mathcal{F}_t \not\subset \mathcal{E}$ for some $t$, then there exists an agent strategy for some agent  type $\theta_0 \in \Theta_0$ that has positive expected profit for the agent. That is, the statistical contract is not incentive-aligned.
\end{prop}

As in the single-round setting, incentive-aligned statistical contracts are maximin optimal, as stated next.
\begin{theorem}[Maximin if and only if incentive-aligned]
\label{thm:maximin_multi}
Recall $Q$ denotes a distribution over agent types $\theta$. The following inequality holds if and only if $\mathcal{F}_1,\dots,\mathcal{F}_T$ yield an incentive-aligned statistical contract:
\begin{multline}
        \inf_{Q}\left\{\E_{\theta \sim Q}\left[ \E\left[ u\left( \theta, L + P \right) \mid \theta; {\textsc{br-strategy}}\left( \theta ; \mathcal{F}_1,\dots,\mathcal{F}_T \right) \right] \right] \right\} \\ = \sup_{\mathcal{F}_1',\dots,\mathcal{F}_T'}\inf_{Q}\left\{\E_{\theta \sim Q}\left[ \E\left[ u\left( \theta, L + P \right) \mid \theta; {\textsc{br-strategy}}\left( \theta ; \mathcal{F}_1',\dots,\mathcal{F}_T' \right) \right] \right] \right\}.
\end{multline}
\end{theorem}
The notation ``${\textsc{br-strategy}}\left( \dots \right)$'' is used to indicate that the agent acts according to a best-response strategy, given their knowledge of their type~$\theta$.

As in the single-round setting, we again see that larger menus are better for the agent, so making larger menus will maximize agent participation and the expected profit cap. That is, the analog of Proposition~\ref{prop:max_participation} holds in the multi-round setting as well. We omit a formal statement for brevity. In addition, when the principal's utility is linear in $L$ for nonnull $\theta$, then the menu of all $e$-values is the best incentive-aligned statistical contract, as stated next.

\begin{prop}[Optimality of the menu of all incentive-aligned statistical contracts under linear utility]
\label{prop:opt_linear_util_multi}
Consider the multi-round setting above. Assume $u\left( \theta_1, L \right)$ is affine in $L$ for $\theta_1 \in \Theta_1$, and let $Q$ be any distribution over $\Theta$. Let $\mathcal{F}_t = \cost_t \cdot \mathcal{E}$ for $t = 1,\dots,T$, and let $\mathcal{F}_t'$ for $t=1,\dots,T$ be any other menu of license functions such that the statistical contract is incentive-aligned. Then,
\begin{align*}
        \E_{\theta \sim Q}&\left[  \E\left[ u\left( \theta, L + P \right) \mid \theta; {\textsc{br-strategy}}\left( \theta ; \mathcal{F}_1,\dots,\mathcal{F}_T \right) \right] \right] \\ &\ge \E_{\theta \sim Q}\left[ \E\left[ u\left( \theta, L + P \right) \mid \theta; {\textsc{br-strategy}}\left( \theta ; \mathcal{F}_1',\dots,\mathcal{F}_T' \right) \right] \right].
\end{align*}
\end{prop}

From Theorem~\ref{thm:maximin_multi} and Proposition~\ref{prop:opt_linear_util_multi}, we conclude that for a principal without knowledge of the distribution of agent types $\theta$, a good choice of statistical contract is to take the menu of all $e$-values at each step: $\mathcal{F}_t = \mathcal{E}$ for $t=1,\dots,T$.

\subsection{Martingale underpinnings}
We have seen above that the multi-round contract with $e$-value updates is incentive-aligned: a null agent has expected value that is less than or equal to their total investment. The main mathematical ingredient underlying this result is the construction of a random process that is a supermartingale for null agents. In particular, we consider the net profit
\begin{equation}
\label{eq:net_profit_mgale}
    N\left( t \right) = L\left( t \right) + \overline{P}\left( t \right) - \overline{\cost}\left( t \right)
\end{equation}
for $t=1,\dots,T$, where $\overline{P}\left( t \right) = P\left( 1 \right) + \dots + P\left( t \right)$ is the total realized profit up to time $t$ and $\overline{\cost}\left( t \right) = \cost_1 I_1 + \dots, \cost_t I_t$ is the total investment up to time $t$. $N\left( t \right)$ represents net profit, since the agent has a license to make $L\left( t \right)$ in future profits, has made $\overline{P}\left( t \right)$ profit thus far, and has invested $\overline{\cost}\left( t \right)$ thus far. The net profit $N\left( t \right)$ is a supermartingale when $\theta \in \Theta_0$. The optional stopping theorem then implies that there is no agent strategy that would result in positive expected total profit.  This is Proposition~\ref{prop:$e$-values-ia-contract_multi}.
In sum, we construct a random process tracking the total expected profit, and use martingale results to control its expectation for all agent strategies.

\subsection{Agent's best response via dynamic programming}\label{sec:dynamic-programming}

We show how the agent can identify their optimal choice of update function $f_t$ in the multi-round contract by solving a dynamic programming problem in Appendix~\ref{app:multiround_policy}. In particular, we show how to choose $f_t$ at each stage to optimize their final license value after at most~$T$ rounds.

There are three main ingredients. First, we show how to find the optimal choice of $f_t$ in a single round when the agent's objective function is a concave, increasing value function. Second, we argue that it is sufficient to consider only concave, increasing value functions. Lastly, we show how to find the optimal choice in the multi-round setup by dynamic programming. We formulate the multi-round statistical contract as a Markov decision process (MDP), where the agent's policy specifies the license updates $f_t$ and stopping criteria based on the history of the trial so far. We start at the final time point and proceed backward, and the main operation at each time point is computing the optimal single-step update~$f_t$ with respect to a value function determined by previous steps of the algorithm.

When the agent uses a simple Gaussian null~$Z_t\sim \mathcal{N}\left( 0, 1 \right)$ against a simple Gaussian alternative~$Z_t\sim \mathcal{N}\left( \theta_1, 1 \right)$, we show how 
 to explicitly solve the dynamic program after discretizing the license to take on finitely many values in Appendix~\ref{app:multiround_policy}. We use this result to derive an optimal license using dynamic programming for a multi-round agent in a trial with at most 5 rounds, where the price to observe each round (including the first) is $\cost_t = 0.1$. We perform a simulation study to compare the multi-round agent's results to two different one-period agents. The first one-round agent observes~$Z\sim \mathcal{N}\left( \theta_1, 1 \right)$ and pays~$\cost = 0.1$ to run the trial; the second one-round agent has~$5\times$ the amount of data, so $Z\sim \mathcal{N}\left( \theta_1, \frac{1}{\sqrt{5}} \right)$ and pays~$5\cost = 0.5$ to run the trial. For both of these agents, we use the optimal license from Section~\ref{sec:simple_example}.

Figure~\ref{fig-dp1} summarizes the results of our simulation study. The first two panels (a-b) of Figure~\ref{fig-dp1} use a specific alternative value~$\theta_1 = 1.645$. In this case, both the one period agent with~$5\times$ the amount of data and the $T=5$ period agent reach a license value equal to their pre-specified cap with probability close to one. The difference, however, is that the one-period agent pays~$5\cost = .5$ to run the trial, whereas multi-round agent pays~$\tau \cost$, where~$\tau$ is the number of rounds ran, which in this example was less than~$5$ in most cases. In the latter two panels (c-d), we find for a variety of alternatives that the multi-round agent can leverage the cost savings to consistently achieve a comparable or in some cases significantly higher overall profit.

\begin{figure}
	\centering
	\centerline{
	\includegraphics[width=.5\textwidth]{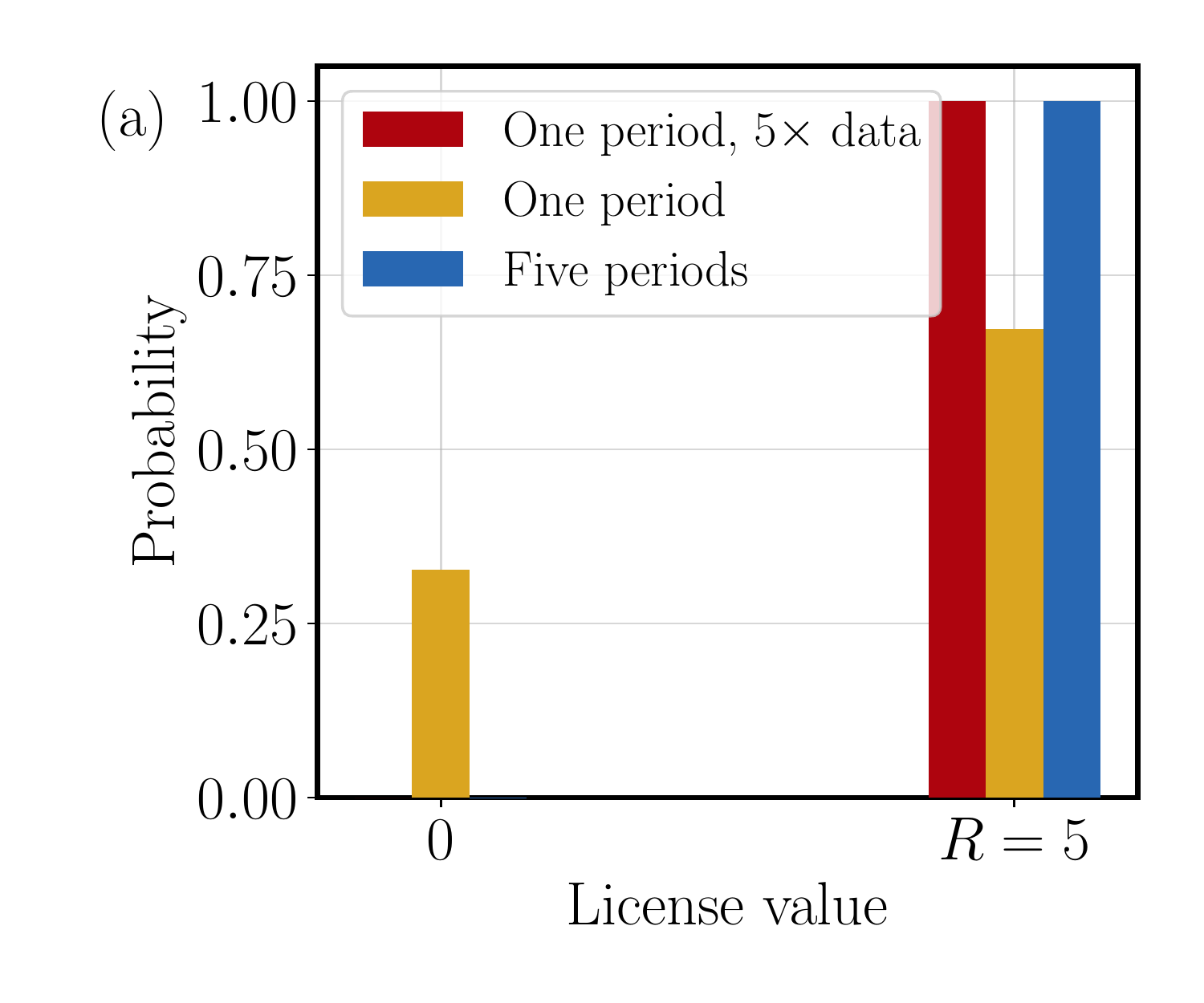}
	\includegraphics[width=.5\textwidth]{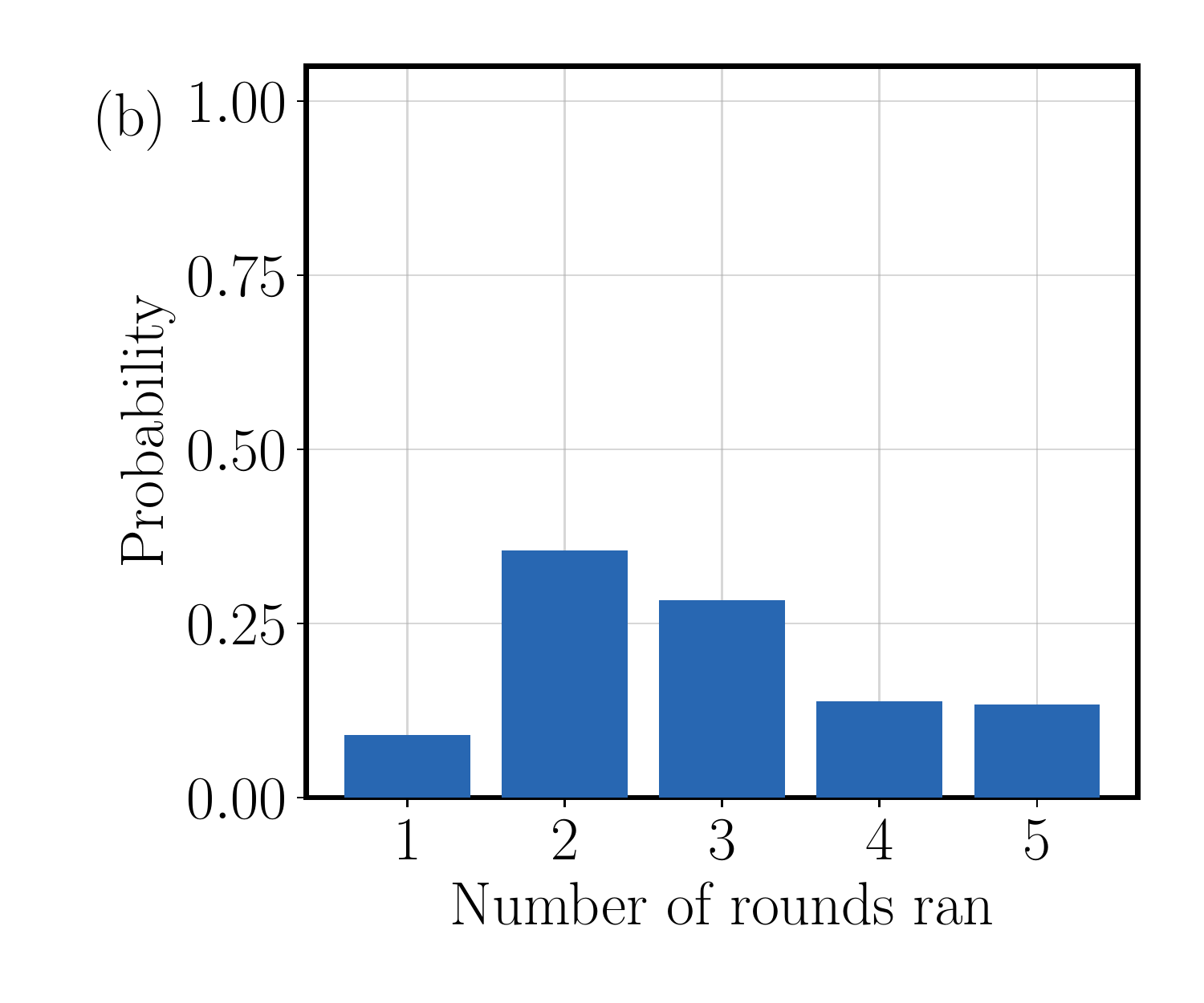}
	}

     \centerline{
	\includegraphics[width=.5\textwidth]{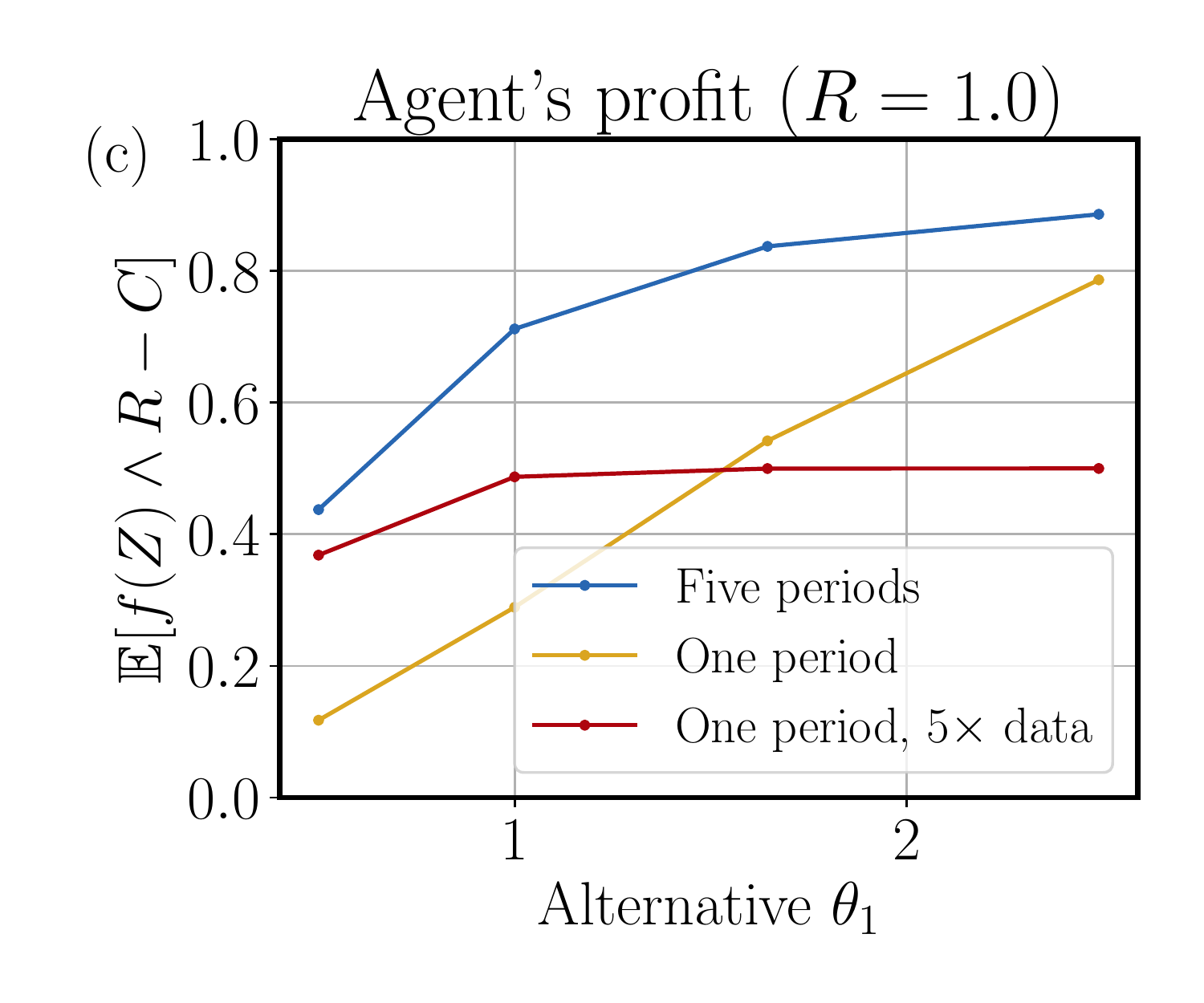}
	\includegraphics[width=.5\textwidth]{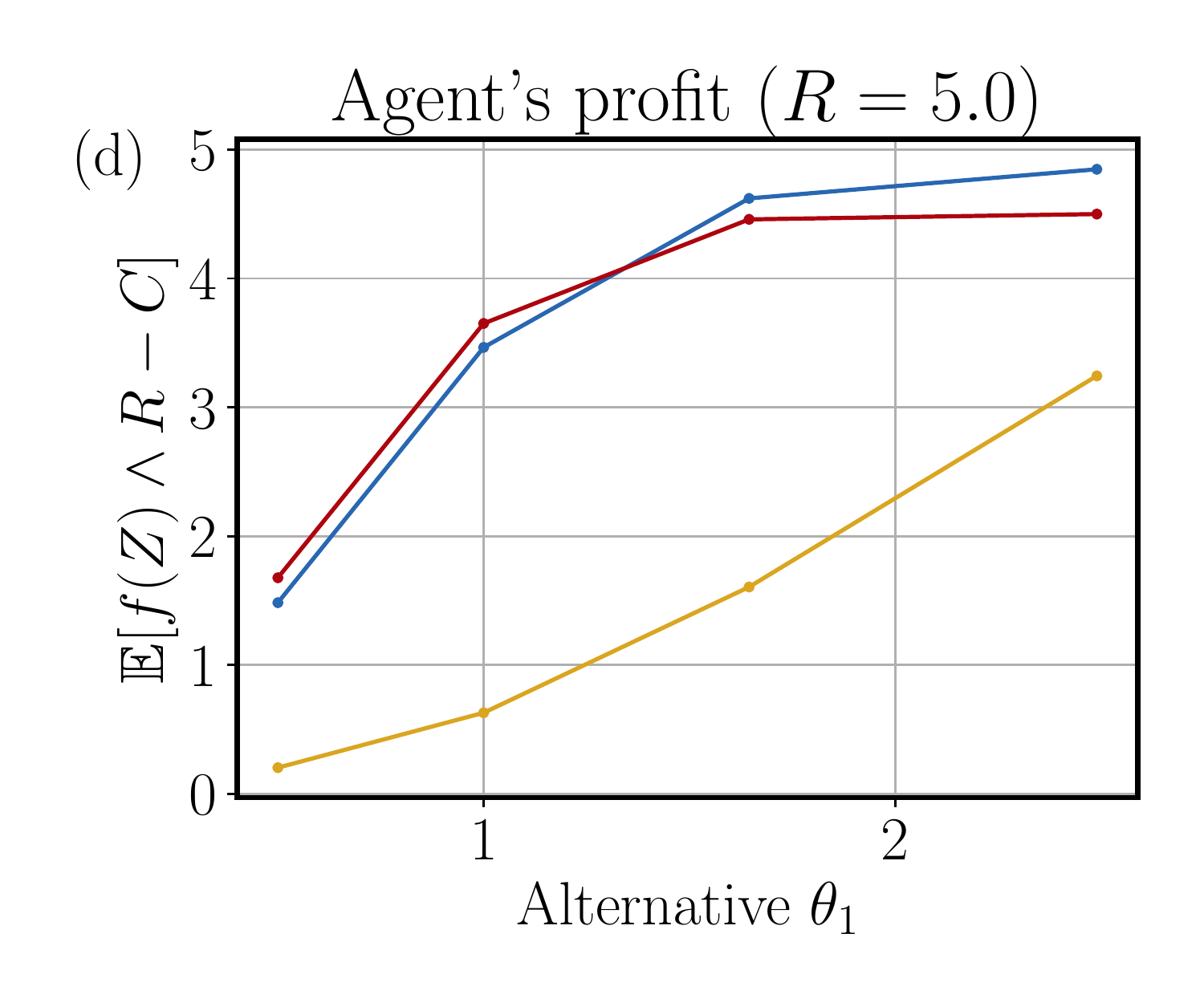}
	}

    \caption{Simulation study comparing a multi-round statistical contract with two one-round contracts. See Section~\ref{sec:dynamic-programming} for a detailed description of each agent's strategy. (a) The terminal license value for each strategy. (b) The distribution of the number of rounds used by the multi-period agent. (c-d) The agent's profit as a function of alternative~$\theta_1$ with a maximal profit of~$\marketsize=1$ and~$\marketsize=5$, respectively.}\label{fig-dp1}

\end{figure}

\subsection{Deriving the agent's best response}
\label{app:multiround_policy}

In this section we show how the agent may derive their optimal multi-round statistical contract. In Section~\ref{sec-mdp} we formulate a Markov Decision Process (MDP) where the update functions $\left( f_t \right)$ correspond to the agent's policy. At each iteration of dynamic programming, the agent must find an optimal $e$-value with a general objective function. We first find the optimal single-round update in Section~\ref{sec-opt-general}.

\subsubsection{Computing the optimal single-round \texorpdfstring{$e$}{e}-value}\label{sec-opt-general}

Consider an agent seeking to maximize $\E_{\theta_1}\left[ v\left( f\left( Z \right) \right) \right]$ over the function $f$ subject to $\E_0\left[ f \right]\le \cost$, for some value function $v : \mathbb{R}_+ \to \mathbb{R}_+$. We first show that we may assume that~$v$ is nondecreasing and concave with no loss of generality. 

\subsubsection{Reduction to concave increasing utilities}

Without loss of generality we may assume $v$ is concave. Why? Because if we define the least concave majorant $\bar{v} = \textnormal{LCM}\left( v \right)$, no optimal $f$ can have probability mass on the set
$$A := \left\{ Z: v\left( f\left( Z \right) \right) < \bar{v}\left( f\left( Z \right) \right) \right\}.$$

This would cause a contradiction, as for any such $Z$ in this set, $v\left( f\left( Z \right) \right)$ is below its concave majorant and therefore
$$v\left( f\left( Z \right) \right) < a v\left( f^- \right) + \left( 1-a \right) v\left( f^+ \right),$$

where $f^-$, $f^+$ are nodes on the concave majorant, and $f\left( Z \right) = a f^- + \left( 1-a \right)f^+$. This equation implies that the value can be strictly improved by taking a weighted combination of these two nodes with probability $a$, with no change to the constraint $\E_0\left[ f \right]\le \cost$.
 
Consequently, all optimal solutions under value function $v$ occur where $\bar{v} = v$ and the set $A$ is empty. And, as $\bar{v}$ dominates~$v$ pointwise, its optimal value is at least as large as the optimal value achievable under $v$.

We shall now solve the problem for $\bar{v}$. By a similar logic we can show that for any optimal solution $\bar{f}$ under $\bar{v}$, there exists another solution $\tilde{f}$ which is fully supported on $A^\mathsf{c}$. This solution $\tilde{f}$ is therefore also a valid solution for $v$, and therefore an optimal solution for $v$. 

Without loss of generality we may also assume $\bar{v}$ is nondecreasing: defining the function
$$\tilde{v}\left( x \right) := \sup \limits_{y \le x} \bar{v}\left( y \right),$$
it becomes clear that any solution  under value function $\bar{v}$ which places positive mass on the set 
$$B := \left\{ Z: \bar{v}\left( f\left( z \right) \right) \neq \tilde{v}\left( f\left( z \right) \right) \right\}$$ can be trivially improved by reducing $f$ at some values, which makes the constraint $\E_0\left[ f \right]\le \cost$ easier to achieve.

So it suffices to restrict ourselves to solutions of $\tilde{v}$, and specifically those which limit to the leftmost point within each flat range of $\bar{v}$.

Thus, $\tilde{v}$ is nondecresing and concave: we proceed to solve it in a simple Gaussian location family, and show how to derive an optimal solution supported only on the nodes.

\subsubsection{Maximization under concave and nondecreasing value function}
Consider maximizing over $f \in \mathcal{F}$ the following:
\[
\int \left( p_{\theta_1}v\left( f \right) - \lambda p_{\theta_0}f \right) \textnormal{d}\mu.
\]
Pointwise maximization yields
\[
f^{\textnormal{br}}\left( z \right) \coloneqq {\arg\max}_{a\in \mathbb{R}_+}\, p_{\theta_1}\left( z \right)v\left( a \right) - \lambda p_{\theta_0}\left( z \right)a,
\]
where $\lambda > 0$ is chosen so that $\mathbb{E}_0\left[ f^{\textnormal{br}}\left( Z \right) \right] = \cost$.

Concretely, consider a Gaussian location model with~$\theta_0 = 0$ and $\theta_1 = \theta$. The likelihood ratio is 
\[
L =\exp\left( \theta Y - \frac{\theta^2}{2} \right).
\]
Note that
\begin{align*}
\mathbb{E}_0\left[ E_{\lambda}\left( L \right) \right]
&= \mathbb{E}_0\left[\max\left\{x : \partial_-v\left( x \right)\ge \frac{\lambda}{L}\right\}\right].
\end{align*}
If $v$ is continuously differentiably, 
\begin{align*}
\mathbb{E}_0\left[ E_{\lambda}\left( L \right) \right]
&= \mathbb{E}_0\left[\left( v' \right)^{-1}\left( \lambda\exp\left( \frac{\theta^2}{2} - Y\theta \right) \right)\right].
\end{align*}
For example, if $v\left( x \right) = 2\sqrt{x}$ then $v'\left( x \right) = x^{-1/2}$ and $\left( v' \right)^{-1}\left( t \right) = t^{-2}$ so
\begin{align*}
\mathbb{E}_0\left[ E_{\lambda}\left( L \right) \right]
&= \lambda^{-2}e^{-\theta^2}\mathbb{E}_0\left[e^{2Y\theta}\right] \\
&= \lambda^{-2}e^{\theta^2},
\end{align*}
yielding $\lambda^* = \exp\left( \theta^2/2 \right)$, and 
$$
E_{\lambda^*}
= \exp\left( 2Y\theta - 2\theta^2 \right).
$$
Furthermore, the optimal value under the alternative is 
$$
\mathbb{E}_\theta\left[v\left( E_{\lambda^*} \right)\right]
= 2\mathbb{E}_\theta\left[\exp\left( Y\theta - \theta^2 \right)\right]
= 2\exp\left( \theta^2/2 \right).
$$

If $v$ is piecewise linear and bounded, with knots $\left( x_k \right)_{k=1}^K$ and left-hand slopes $\left( v'_k \right)_{k=1}^K$ (prepending $x_0=0$ and $v'_0 = \infty$), we have
\begin{align*}
\mathbb{E}_0\left[ E_{\lambda}\left( L \right) \right]
&= \mathbb{E}_0\left[\max\left\{x_k : v'_k\ge \frac{\lambda}{L}\right\}\right] \\
&= \mathbb{E}_0\left[\max\left\{x_k : v'_k\ge \lambda\exp\left( \frac{\theta^2}{2} - Y\theta \right)\right\}\right] \\
&= \int_{-\infty}^\infty \phi\left( y \right)\left[\max\left\{x_k : v'_k\ge \lambda\exp\left( \frac{\theta^2}{2} - y\theta \right)\right\}\right]\text{d}y \\
&= \sum_{k=0}^K x_k\int_{A_k} \phi\left( y \right)\text{d}y,
\end{align*}
where 
\begin{align*}
    A_k 
&= \left\{y  : x_k = \max\left\{x_l : v'_l\ge \lambda\exp\left( \frac{\theta^2}{2} - y\theta \right)\right\}\right\} \\
&= \left\{y  : x_k = \max\left\{x_l : y\ge \frac{\theta}{2} - \frac{1}{\theta}\log\left( \frac{v'_l}{\lambda} \right)\right\}\right\}.
\end{align*}
Hence $A_0, \ldots, A_K$ are given by
\begin{align*}
A_0 &= \left( -\infty, \frac{\theta}{2} - \frac{1}{\theta}\log\left( \frac{v'_1}{\lambda} \right)\right] \\
A_k &= \left( \frac{\theta}{2} - \frac{1}{\theta}\log\left( \frac{v'_k}{\lambda} \right), \frac{\theta}{2} - \frac{1}{\theta}\log\left( \frac{v'_{k+1}}{\lambda} \right)\right] \\
A_K &= \left( \frac{\theta}{2} - \frac{1}{\theta}\log\left( \frac{v'_K}{\lambda} \right), \infty \right).
\end{align*}
Finally, 
\begin{align*}
\mathbb{E}_0\left[ E_{\lambda}\left( L \right) \right]
&= x_K\left( 1-\Phi\left( \frac{\theta}{2} - \frac{1}{\theta}\log\left( \frac{v'_{K}}{\lambda} \right) \right) \right) \\
&~~~~~~+ \sum_{k=1}^{K-1} x_k\left( 
\Phi\left( \frac{\theta}{2} - \frac{1}{\theta}\log\left( \frac{v'_{k+1}}{\lambda} \right) \right)
- \Phi\left( \frac{\theta}{2} - \frac{1}{\theta}\log\left( \frac{v'_{k}}{\lambda} \right) \right).
 \right)
\end{align*}
This is monotone decreasing in $\lambda$ so solve for $\lambda^*$ via binary search.

The expected value under the alternative is
\begin{align*}
\mathbb{E}_\theta\left[ E_{\lambda}\left( L \right) \right]
&= x_K\left( 1-\Phi\left( -\frac{\theta}{2} - \frac{1}{\theta}\log\left( \frac{v'_{K}}{\lambda} \right) \right) \right) \\
&~~~~~~+ \sum_{k=1}^{K-1} x_k\left( 
\Phi\left( -\frac{\theta}{2} - \frac{1}{\theta}\log\left( \frac{v'_{k+1}}{\lambda} \right) \right)
- \Phi\left( -\frac{\theta}{2} - \frac{1}{\theta}\log\left( \frac{v'_{k}}{\lambda} \right) \right)
 \right).
\end{align*}
The optimal objective under the alternative is
\begin{multline}
\mathbb{E}_\theta\left[ v\left( E_{\lambda^*}\left( L \right) \right) \right]
= v\left( x_K \right)\left( 1-\Phi\left( -\frac{\theta}{2} - \frac{1}{\theta}\log\left( \frac{v'_{K}}{\lambda} \right) \right) \right) 
\\ + \sum_{k=1}^{K-1} v\left( x_k \right)\left( 
\Phi\left( -\frac{\theta}{2} - \frac{1}{\theta}\log\left( \frac{v'_{k+1}}{\lambda} \right) \right)
- \Phi\left( -\frac{\theta}{2} - \frac{1}{\theta}\log\left( \frac{v'_{k}}{\lambda} \right) \right)
 \right).
\end{multline}

\subsubsection{Formulating the Markov decision process}\label{sec-mdp}

We may extend the approach of the previous section on one-period optimization with concave utility to construct the optimal $e$-value sequence for the multi-period setting.

We shall assume an exact point alternative $\theta_a$ is known (Note: the methods below can in fact be extended to distributions $P_a$ on $\theta_a$ as well, but we must adjoin the evolving sufficient statistics to the Markov Decision Process (MDP) state space---we omit this for simplicity.) If we consider the family of licenses which can take at most $N$ distinct values (noting that such licenses will be effectively dense in the relevant space of all valid licenses as we take arbitrarily large $N$), this multi-round process is a MDP with state space given by the current time and license value action space, $\left( t, L\left( t \right) \right)$, given by either stopping or the selection of an $e$-value $f_{t+1}\left( \cdot \right)$ to represent the next multiplicative evolution of the license, and reward given in the final stage as $u\left( L\left( \tau \right), \sum \limits_{t=1}^\tau \cost_t \right)$. 
(The transition probabilities are implied by the distribution of license values, $L\left( t+1 \right) \sim P_{a,t}$.)

Consequently, this MDP admits a dynamic programming solution:

$$v_{i-1}\left( s \right) := \max \left\{u\left( L\left( i-1 \right), \sum \limits_{t=1}^{i-1} \cost_t  \right) , \max \limits_a \left[ \sum \limits_{s'} P_a\left( s,s' \right) v_{i}\left( s' \right)  \right]\right\},$$
where $v_T\left( s \right) = u\left( L\left( T \right),\sum \limits_{t=1}^T \cost_t \right)$ initializes the backward induction. Note that the problem of finding the correct action $a$ is equivalent to the one-period optimization problem, with utility function determined by $v_i\left( s' \right)$. Thus, we will boil down the multi-period problem to a sequence of such one-dimension problems: we must alternate between the two steps: (a) find the optimal one-period $e$-value (action) under $\left( i, l \right)$ for a collection of values $l$, and (b) determine the next iteration of the value function by finding expected utility of that one-period license, to evaluate $v_i\left( s \right)$, and computing the appropriate maximum. To achieve (a) practically, we discretize by constraining ~$f_t$ such that $L\left( t \right) \in \mathcal{D} = \left\{ 0, \varepsilon, \ldots, K\varepsilon \right\}$ for each $t$. As a result, we are able to compute a full optimal strategy among licenses which are supported on $\mathcal{D} = \left\{ 0, \varepsilon, \ldots, K\varepsilon \right\}$ at all times. Finally, $\epsilon$ is chosen small to ensure near-optimal performance.

\end{document}